\documentclass[11pt]{article}
\usepackage{amsmath,amsfonts,amssymb,amsthm,mathtools,titling,url,array,algorithm}
\usepackage[shortlabels]{enumitem}
\usepackage{fullpage}
\usepackage{graphicx}
\usepackage{tikz}

\usepackage[font=small,labelfont=bf]{caption}

\usetikzlibrary{automata,arrows,positioning,graphs,graphs.standard,decorations.markings,decorations.pathreplacing}

\tikzset{
    fo/.style={ 
        decoration={
            markings,
            mark=at position #1 with {\arrow[line width=0.9mm]{>}}
        },
        postaction={decorate}
    },
    fo/.default=0.58,
    ro/.style={ 
        decoration={
            markings,
            mark=at position #1 with {\arrow[line width=0.9mm]{<}}
        },
        postaction={decorate}
    },
    ro/.default=0.58,
}
\definecolor{tikzblue}{HTML}{1C71D8}

\definecolor{ForestGreen}{rgb}{0.0333,0.4451,0.0333}
\definecolor{DarkRed}{rgb}{0.65,0,0}
\definecolor{Red}{rgb}{1,0,0}
\usepackage[linktocpage=true,
pagebackref=true,colorlinks,
linkcolor=DarkRed,citecolor=ForestGreen,
bookmarks,bookmarksopen,bookmarksnumbered]{hyperref}

\usepackage{cleveref}
\usepackage{thm-restate}
\usepackage[normalem]{ulem}
\usepackage{algorithm}
\usepackage{xfrac}

\theoremstyle{plain}
\newtheorem{theorem}{Theorem}[section]
\newtheorem{lemma}[theorem]{Lemma}

\newtheorem{corollary}[theorem]{Corollary}

\newtheorem{definition}[theorem]{Definition}

\newtheorem{problem}[theorem]{Problem}

\newcommand{\lca}{\text{lca}}
\newcommand{\drop}{\text{drop}}
\newcommand{\slack}{\text{slack}}

\newcommand{\vecF}{\Vec{F}}
\newcommand{\vecJ}{\Vec{J}}
\newcommand{\OPT}{\text{OPT}}

\newcommand*\mc[1]{\mathcal{#1}}
\DeclarePairedDelimiter\abs{\lvert}{\rvert}
\DeclarePairedDelimiter\norm{\lVert}{\rVert}%
\makeatletter
\let\oldabs\abs
\def\abs{\@ifstar{\oldabs}{\oldabs*}}
\let\oldnorm\norm
\def\norm{\@ifstar{\oldnorm}{\oldnorm*}}
\makeatother

\usepackage[colorinlistoftodos,prependcaption,textsize=tiny]{todonotes}

\definecolor{dnotecol}{rgb}{0.20, 0.50, 0.80}

\title{Approximation Algorithms for Steiner Connectivity Augmentation}

\author{Daniel Hathcock \and Michael Zlatin}

\begin{document}

\maketitle

\begin{abstract} 
We consider connectivity augmentation problems in the Steiner setting, where the goal is to augment the edge-connectivity between a specified subset of terminal nodes.

In the Steiner Augmentation of a Graph problem ($k$-SAG), we are given a $k$-edge-connected subgraph $H$ of a graph $G$. The goal is to augment $H$ by including links from $G$ of minimum cost so that the edge-connectivity between nodes of $H$ increases by 1. This is a generalization of the Weighted Connectivity Augmentation Problem, in which only links between pairs of nodes in $H$ are available for the augmentation. 

In the Steiner Connectivity Augmentation Problem ($k$-SCAP), we are given a Steiner $k$-edge-connected graph connecting terminals $R$, and we seek to add links of minimum cost to create a Steiner $(k+1)$-edge-connected graph for $R$. Note that $k$-SAG is a special case of $k$-SCAP.



The results of Ravi, Zhang and Zlatin for the Steiner Tree Augmentation problem yield a $(1.5+\varepsilon)$-approximation for $1$-SCAP and for $k$-SAG when $k$ is odd~\cite{RZZ2022}. In this work, we give a $(1 + \ln{2} +\varepsilon)$-approximation for the Steiner Ring Augmentation Problem (SRAP). This yields a polynomial time algorithm with approximation ratio $(1 + \ln{2} + \varepsilon)$ for $2$-SCAP. We obtain an improved approximation guarantee for SRAP when the ring consists of only terminals, yielding a $(1.5+\varepsilon)$-approximation for $k$-SAG for any $k$. 

 
 
\end{abstract}\thispagestyle{empty}

\newpage
\tableofcontents

\newpage
\setcounter{page}{1}

\section{Introduction}
\subsection{Background}
The edge-connectivity of a graph is a common measure of the robustness of the network to edge failures. If a network is $k$-edge-connected, then it can sustain failures of up to $k-1$ edges without being disconnected. Many problems in the area of network design seek to construct a cheap network which satisfies edge-connectivity requirements between certain pairs of nodes. 

This has given rise to many fundamental problems of interest in combinatorial optimization and approximation algorithms. One notable example is the Survivable Network Design Problem (SNDP). In SNDP, we are given a graph with non-negative costs on edges and a connectivity requirement $r_{ij}$ for each pair of vertices $i,j \in V$. The goal is to find a cheapest subgraph of $G$ so that there are $r_{ij}$ pairwise edge-disjoint paths between all pairs of vertices $i$ and $j$. 

Jain~\cite{J2001} gave a polynomial time algorithm for SNDP based on iterative rounding, which achieves an approximation factor of 2. Despite its generality, this algorithm achieves the best-known approximation ratio for a variety of network design problems of particular interest. For example, when $r_{ij} = k$ for all $i,j \in V$, we obtain the weighted minimum cost $k$-edge-connected spanning subgraph problem, for which no better-than-2 approximation algorithm is known, even when $k=2$. However, for some special cases, such as the minimum Steiner tree problem, specialized algorithms have been developed to improve upon this ratio. Hence, a major question in the field of network design is: for which network connectivity problems can we achieve an approximation factor below 2?

Recently, there has been major progress towards addressing this question. Consider the special case of SNDP known as the Weighted Connectivity Augmentation Problem (WCAP), in which we seek to increase the edge-connectivity of a given $k$-edge-connected graph by 1 by adding to the graph a collection of links of minimum cost.


A solution to WCAP must include a link crossing each of the min-cuts of the given graph $H$. Since the minimum cuts of any graph can be represented by a cactus~\cite{FF2009}, it is enough to consider the WCAP problem for $k = 2$ and when $H$ is a cactus.\footnote{A cactus is a connected graph in which every edge is included in exactly 1 cycle.} 
In fact, by the addition of zero cost links, the weighted problem further reduces to the case where $H$ is a cycle (the so-called Weighted Ring Augmentation Problem, WRAP)~\cite{galvez2021cycle}. If $k$ is odd, WCAP reduces to the case where $k = 1$ and $H$ is a tree, yielding the Weighted Tree Augmentation Problem (WTAP).

In a breakthrough result, Traub and Zenklusen~\cite{TZ2021} employed a relative greedy algorithm using ideas from~\cite{CN2013} to give the first approximation algorithm with approximation ratio better than 2 for WTAP. In particular, they achieved an approximation ratio of $(1+ \ln{2} + \varepsilon)$ for WTAP, thereby yielding the same result for WCAP when $k$ is odd. Shortly afterward, they improved upon this result, bringing the approximation factor down to $(1.5+\varepsilon)$ for WTAP~\cite{TZ2022} using a local search algorithm. Finally, they turned to the general case of WCAP, adapting the local search method that was used for WTAP to give a $(1.5 + \varepsilon)$-approximation algorithm for the Weighted Ring Augmentation Problem~\cite{TZ2022new}. As discussed above, this gives a $(1.5+ \varepsilon)$-approximation for general WCAP, adding it to the meager list of NP-hard special cases of SNDP which we can approximate to within a factor less than 2.
 

This brings us to our setting. Note that for the above augmentation problems, the parameter of interest is the \textit{global} edge-connectivity of the graph. However, in many applications, we are not interested in connectivity between all nodes in the graph, but rather only the connectivity between a specified subset of ``important" nodes called terminals. 



The reasons to be interested in Steiner connectivity are twofold. The first is that utilizing Steiner vertices outside a given network may allow us to augment it more cheaply. For example, in \Cref{fig:reasonforSAG}, we seek to augment the edge-connectivity of the nodes in a tree from 1 to 2. The ability to use a Steiner vertex outside the tree results in a cheaper augmentation.

    

\begin{figure}[h]
    \centering
    \includegraphics[scale = 1.1]{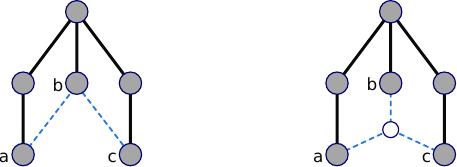}
    
        \caption{We seek to augment the edge-connectivity between the grey nodes (terminals) in the given tree from 1 to 2. Nodes $a$, $b$ and $c$ are arranged in an equilateral unit triangle. If only direct links are allowed, the cheapest augmentation has cost 2, but utilizing a Steiner node in the center of this triangle allows us to establish 2-edge-connectivity between the terminals at a cost of $\sqrt{3} \approx 1.73$.}
    \label{fig:reasonforSAG}
\end{figure}

The second reason is that we may want to augment the resiliency of a network which already incorporates Steiner vertices. However, we only desire to increase the edge-connectivity between its terminals. This scenario is likely in practice since, as we have seen, using Steiner nodes can establish connectivity while maintaining low costs. Naively, attempting to accomplish this with an algorithm for global connectivity augmentation may result in a needlessly expensive solution. See \Cref{fig:reasonforSCAP} for an illustration. 

\begin{figure}[h]
    \centering
    \includegraphics[scale = 1.1]{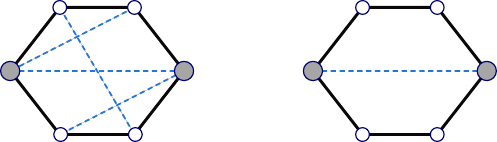}
    \caption{The given graph has 2-edge-connectivity between the large grey nodes (the terminals), which we seek to increase to 3-edge-connectivity. Solving this with global connectivity augmentation results in a solution with 4 links when only 1 is required.}
    \label{fig:reasonforSCAP}
\end{figure}

In this paper, we develop algorithms for generalizations of WCAP to this Steiner setting. We first consider the Steiner Augmentation of a Graph problem (SAG). In contrast to WCAP, in which we can only add links between vertices of the graph to be augmented, in the SAG problem we can purchase links that are incident to Steiner nodes outside the given graph.

\begin{problem}[$k$-Steiner Augmentation of a Graph]
We are given a $k$-edge-connected graph $H = (R, E)$, which is a subgraph of $G = (V,E~ \dot \cup ~L)$. The links $L$ have non-negative costs $c: L \to \mathbb{R}_{\geq 0}$.

The goal is to select $S \subseteq L$ of cheapest cost so that the graph $H' = (V, E \cup S)$ has $k+1$ pairwise edge-disjoint paths between $u$ and $v$ for all $u,v \in R$.
\end{problem}

More generally, we define the Steiner Connectivity Augmentation Problem (SCAP). The goal is to augment the edge-connectivity between a specified subset of terminals in a graph from $k$ to $k+1$ in the cheapest way. 

\begin{problem}[$k$-Steiner Connectivity Augmentation Problem]
We are given a graph $G = (V, E~ \dot\cup ~L)$ and a subset of terminals $R \subseteq V$ such that $H := (V,E)$ has $k$ edge-disjoint paths between every pair of vertices in $R$. We are also given a cost function $c: L \to \mathbb{R}_{\geq 0}$.

The goal is to select $S \subseteq L$ of cheapest cost so that the graph $(V, E \cup S)$ has $k+1$ pairwise edge-disjoint paths between all pairs of nodes in $R$.
\end{problem}

See \Cref{fig:scap-and-sag} for example instances of SAG and SCAP.

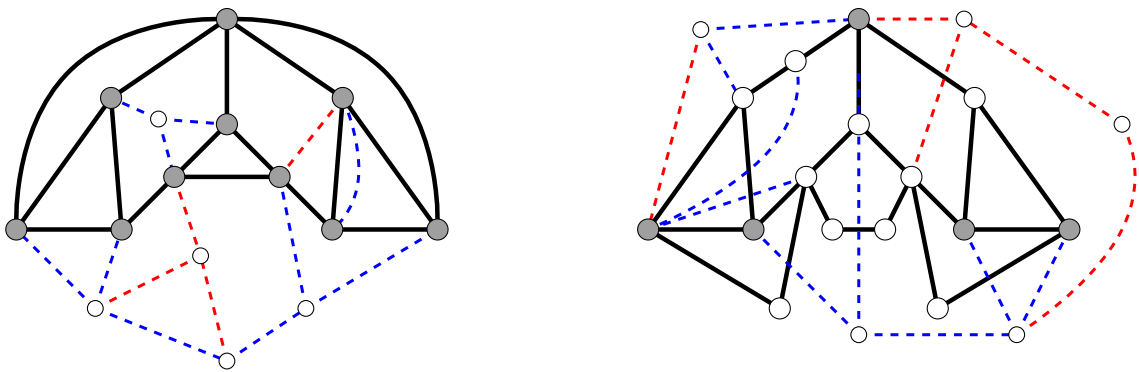
\begin{figure}[h!]
\centering
\begin{tikzpicture}[scale=0.55]

\begin{scope} 
    

    \draw [-] [black,line width=0.6mm,xshift=0 cm] plot [smooth, tension=1] coordinates {(0, 1) (1, 4) (4, 5)};
    \draw [-] [black,line width=0.6mm,xshift=0 cm] plot [smooth, tension=1] coordinates {(0, 1) (1.8, 3.5)};
    \draw [-] [black,line width=0.6mm,xshift=0 cm] plot [smooth, tension=1] coordinates {(0, 1) (2, 1)};
    \draw [-] [black,line width=0.6mm,xshift=0 cm] plot [smooth, tension=1] coordinates {(1.8, 3.5) (2, 1)};
    \draw [-] [black,line width=0.6mm,xshift=0 cm] plot [smooth, tension=1] coordinates {(1.8, 3.5) (4, 5)};
    \draw [-] [black,line width=0.6mm,xshift=0 cm] plot [smooth, tension=1] coordinates {(2, 1) (3, 2)};
    \draw [-] [black,line width=0.6mm,xshift=0 cm] plot [smooth, tension=1] coordinates {(4, 3) (3, 2)};
    \draw [-] [black,line width=0.6mm,xshift=0 cm] plot [smooth, tension=1] coordinates {(5, 2) (3, 2)};
    \draw [-] [black,line width=0.6mm,xshift=0 cm] plot [smooth, tension=1] coordinates {(5, 2) (4, 3)};
    \draw [-] [black,line width=0.6mm,xshift=0 cm] plot [smooth, tension=1] coordinates {(4, 5) (4, 3)};
    \draw [-] [black,line width=0.6mm,xshift=0 cm] plot [smooth, tension=1] coordinates {(5, 2) (6, 1)};
    \draw [-] [black,line width=0.6mm,xshift=0 cm] plot [smooth, tension=1] coordinates {(6, 1) (8, 1)};
    \draw [-] [black,line width=0.6mm,xshift=0 cm] plot [smooth, tension=1] coordinates {(6, 1) (6.2, 3.5)};
    \draw [-] [black,line width=0.6mm,xshift=0 cm] plot [smooth, tension=1] coordinates {(8, 1) (6.2, 3.5)};
    \draw [-] [black,line width=0.6mm,xshift=0 cm] plot [smooth, tension=1] coordinates {(4, 5) (6.2, 3.5)};
    \draw [-] [black,line width=0.6mm,xshift=0 cm] plot [smooth, tension=1] coordinates {(8, 1) (7, 4) (4, 5)};

    \draw [dashed] [tikzblue, line width=0.4mm,xshift=0 cm] plot [smooth, tension=1] coordinates {(2.7, 3.1) (1.8, 3.5)};   
    \draw [dashed] [tikzblue, line width=0.4mm,xshift=0 cm] plot [smooth, tension=1] coordinates {(2.7, 3.1) (3, 2)};   
    \draw [dashed] [tikzblue, line width=0.4mm,xshift=0 cm] plot [smooth, tension=1] coordinates {(2.7, 3.1) (4, 3)};   
    \draw [dashed] [tikzblue, line width=0.4mm,xshift=0 cm] plot [smooth, tension=1] coordinates {(1.5, -0.5) (4, -1.5)};   
    \draw [dashed] [tikzblue, line width=0.4mm,xshift=0 cm] plot [smooth, tension=1] coordinates {(1.5, -0.5) (0, 1)};   
    \draw [dashed] [tikzblue, line width=0.4mm,xshift=0 cm] plot [smooth, tension=1] coordinates {(1.5, -0.5) (2, 1)};   
    \draw [dashed] [red, line width=0.4mm,xshift=0 cm] plot [smooth, tension=1] coordinates {(1.5, -0.5) (3.5, 0.5)};   
    \draw [dashed] [red, line width=0.4mm,xshift=0 cm] plot [smooth, tension=1] coordinates {(3, 2) (3.5, 0.5)};   
    \draw [dashed] [red, line width=0.4mm,xshift=0 cm] plot [smooth, tension=1] coordinates {(4, -1.5) (3.5, 0.5)};   
    \draw [dashed] [tikzblue, line width=0.4mm,xshift=0 cm] plot [smooth, tension=1] coordinates {(4, -1.5) (5.5, -0.5)};   
    \draw [dashed] [tikzblue, line width=0.4mm,xshift=0 cm] plot [smooth, tension=1] coordinates {(5.5, -0.5) (5, 2)};   
    \draw [dashed] [tikzblue, line width=0.4mm,xshift=0 cm] plot [smooth, tension=1] coordinates {(5.5, -0.5) (8, 1)};   
    \draw [dashed] [red, line width=0.4mm,xshift=0 cm] plot [smooth, tension=1] coordinates {(6.2, 3.5) (5, 2)};   
    \draw [dashed] [tikzblue, line width=0.4mm,xshift=0 cm] plot [smooth, tension=1] coordinates {(6.2, 3.5) (6.5, 2) (6, 1)};

    \draw[black,fill=white] (2.7, 3.1) ellipse (0.15 cm  and 0.15 cm);
    \draw[black,fill=white] (1.5, -0.5) ellipse (0.15 cm  and 0.15 cm);
    \draw[black,fill=white] (4, -1.5) ellipse (0.15 cm  and 0.15 cm);
    \draw[black,fill=white] (3.5,0.5) ellipse (0.15 cm  and 0.15 cm);
    \draw[black,fill=white] (5.5, -0.5) ellipse (0.15 cm  and 0.15 cm);

        \draw[black,fill=gray!75,] (4,5) ellipse (0.2 cm  and 0.2 cm);	
        \draw[black,fill=gray!75] (0, 1) ellipse (0.2 cm  and 0.2 cm);
        \draw[black,fill=gray!75] (2, 1) ellipse (0.2 cm  and 0.2 cm);
        \draw[black,fill=gray!75] (6, 1) ellipse (0.2 cm  and 0.2 cm);
        \draw[black,fill=gray!75] (8, 1) ellipse (0.2 cm  and 0.2 cm);
        \draw[black,fill=gray!75] (3, 2) ellipse (0.2 cm  and 0.2 cm);
        \draw[black,fill=gray!75] (5, 2) ellipse (0.2 cm  and 0.2 cm);
        \draw[black,fill=gray!75] (1.8, 3.5) ellipse (0.2 cm  and 0.2 cm);
        \draw[black,fill=gray!75] (4, 3) ellipse (0.2 cm  and 0.2 cm);
        \draw[black,fill=gray!75] (6.2, 3.5) ellipse (0.2 cm  and 0.2 cm);

\end{scope}

\begin{scope}[xshift=12cm]

    \draw [-] [black,line width=0.6mm,xshift=0 cm] plot [smooth, tension=1] coordinates {(0, 1) (1.8, 3.5)};
    \draw [-] [black,line width=0.6mm,xshift=0 cm] plot [smooth, tension=1] coordinates {(0, 1) (2, 1)};
    \draw [-] [black,line width=0.6mm,xshift=0 cm] plot [smooth, tension=1] coordinates {(1.8, 3.5) (2, 1)};
    \draw [-] [black,line width=0.6mm,xshift=0 cm] plot [smooth, tension=1] coordinates {(1.8, 3.5) (4, 5)};
    \draw [-] [black,line width=0.6mm,xshift=0 cm] plot [smooth, tension=1] coordinates {(2, 1) (3, 2)};
    \draw [-] [black,line width=0.6mm,xshift=0 cm] plot [smooth, tension=1] coordinates {(4, 3) (3, 2)};
    \draw [-] [black,line width=0.6mm,xshift=0 cm] plot [smooth, tension=1] coordinates {(5, 2) (4, 3)};
    \draw [-] [black,line width=0.6mm,xshift=0 cm] plot [smooth, tension=1] coordinates {(4, 5) (4, 3)};
    \draw [-] [black,line width=0.6mm,xshift=0 cm] plot [smooth, tension=1] coordinates {(5, 2) (6, 1)};
    \draw [-] [black,line width=0.6mm,xshift=0 cm] plot [smooth, tension=1] coordinates {(6, 1) (8, 1)};
    \draw [-] [black,line width=0.6mm,xshift=0 cm] plot [smooth, tension=1] coordinates {(6, 1) (6.2, 3.5)};
    \draw [-] [black,line width=0.6mm,xshift=0 cm] plot [smooth, tension=1] coordinates {(8, 1) (6.2, 3.5)};
    \draw [-] [black,line width=0.6mm,xshift=0 cm] plot [smooth, tension=1] coordinates {(4, 5) (6.2, 3.5)};
    \draw [-] [black,line width=0.6mm,xshift=0 cm] plot [smooth, tension=1] coordinates {(3, 2) (3.5, 1)};
    \draw [-] [black,line width=0.6mm,xshift=0 cm] plot [smooth, tension=1] coordinates {(4.5, 1) (5, 2)};
    \draw [-] [black,line width=0.6mm,xshift=0 cm] plot [smooth, tension=1] coordinates {(4.5, 1) (3.5, 1)};

    \draw [-] [black,line width=0.6mm,xshift=0 cm] plot [smooth, tension=1] coordinates {(2.5, -0.5) (3, 2)};
    \draw [-] [black,line width=0.6mm,xshift=0 cm] plot [smooth, tension=1] coordinates {(2.5, -0.5) (0, 1)};

    \draw [-] [black,line width=0.6mm,xshift=0 cm] plot [smooth, tension=1] coordinates {(5.5, -0.5) (5, 2)};
    \draw [-] [black,line width=0.6mm,xshift=0 cm] plot [smooth, tension=1] coordinates {(5.5, -0.5) (8, 1)};

    \draw [dashed] [tikzblue, line width=0.4mm,xshift=0 cm] plot [smooth, tension=1] coordinates {(4, -1) (4,4)};   
    \draw [dashed] [tikzblue, line width=0.4mm,xshift=0 cm] plot [smooth, tension=1] coordinates {(4, -1) (2, 1)};   
    \draw [dashed] [tikzblue, line width=0.4mm,xshift=0 cm] plot [smooth, tension=1] coordinates {(4, -1) (7, -1)};   
    \draw [dashed] [tikzblue, line width=0.4mm,xshift=0 cm] plot [smooth, tension=1] coordinates {(7,-1) (6, 1)};   
    \draw [dashed] [tikzblue, line width=0.4mm,xshift=0 cm] plot [smooth, tension=1] coordinates {(7, -1) (8, 1)};   
    \draw [dashed] [red, line width=0.4mm,xshift=0 cm] plot [smooth, tension=1] coordinates {(7, -1) (9, 1) (9, 3)};   
    \draw [dashed] [red, line width=0.4mm,xshift=0 cm] plot [smooth, tension=1] coordinates {(9,3) (6, 5)}; 
    \draw [dashed] [red, line width=0.4mm,xshift=0 cm] plot [smooth, tension=1] coordinates {(6,5) (4, 5)}; 
    \draw [dashed] [red, line width=0.4mm,xshift=0 cm] plot [smooth, tension=1] coordinates {(5,2) (6, 5)};   
    \draw [dashed] [tikzblue, line width=0.4mm,xshift=0 cm] plot [smooth, tension=1] coordinates {(1,4.8) (1.8, 3.4)}; 
    \draw [dashed] [tikzblue, line width=0.4mm,xshift=0 cm] plot [smooth, tension=1] coordinates {(1,4.8) (4, 5)}; 
    \draw [dashed] [red, line width=0.4mm,xshift=0 cm] plot [smooth, tension=1] coordinates {(1,4.8) (0, 1)}; 

    \draw [dashed] [tikzblue, line width=0.4mm,xshift=0 cm] plot [smooth, tension=1] coordinates {(2.8, 4.2) (2.3, 2.5) (0, 1)}; 

    \draw [dashed] [tikzblue, line width=0.4mm,xshift=0 cm] plot [smooth, tension=1] coordinates {(3, 2) (0, 1)}; 

        \draw[black,fill=gray!75,] (4,5) ellipse (0.2 cm  and 0.2 cm);	
        \draw[black,fill=gray!75] (0, 1) ellipse (0.2 cm  and 0.2 cm);
        \draw[black,fill=gray!75] (2, 1) ellipse (0.2 cm  and 0.2 cm);
        \draw[black,fill=gray!75] (6, 1) ellipse (0.2 cm  and 0.2 cm);
        \draw[black,fill=gray!75] (8, 1) ellipse (0.2 cm  and 0.2 cm);
        \draw[black,fill=white] (3, 2) ellipse (0.2 cm  and 0.2 cm);
        \draw[black,fill=white] (5, 2) ellipse (0.2 cm  and 0.2 cm);
        \draw[black,fill=white] (1.8, 3.5) ellipse (0.2 cm  and 0.2 cm);
        \draw[black,fill=white] (4, 3) ellipse (0.2 cm  and 0.2 cm);
        \draw[black,fill=white] (6.2, 3.5) ellipse (0.2 cm  and 0.2 cm);

        \draw[black,fill=white] (3.5,1) ellipse (0.2 cm  and 0.2 cm);	
        \draw[black,fill=white] (4.5,1) ellipse (0.2 cm  and 0.2 cm);	
        \draw[black,fill=white] (2.5,-0.5) ellipse (0.2 cm  and 0.2 cm);	
        \draw[black,fill=white] (5.5,-0.5) ellipse (0.2 cm  and 0.2 cm);	

        \draw[black,fill=white] (2.8, 4.2) ellipse (0.2 cm  and 0.2 cm);

    \draw[black,fill=white] (4, -1) ellipse (0.15 cm  and 0.15 cm);
    \draw[black,fill=white] (7, -1) ellipse (0.15 cm  and 0.15 cm);
    \draw[black,fill=white] (6, 5) ellipse (0.15 cm  and 0.15 cm);
    \draw[black,fill=white] (9,3) ellipse (0.15 cm  and 0.15 cm);
    \draw[black,fill=white] (1, 4.8) ellipse (0.15 cm  and 0.15 cm);
    
\end{scope}

\end{tikzpicture}
		
\caption{A 3-SAG instance is shown on the left, and a 3-SCAP instance is shown on the right. The shaded nodes are terminals $R$, the black edges denote the edges of $E$ and the dashed edges represent the links $L$. In both pictures, the blue dashed links form a feasible solution.}
\label{fig:scap-and-sag}
\end{figure}

Notice that $k$-SAG is a special case of $k$-SCAP. In~\cite{RZZ2022}, Ravi, Zhang and Zlatin achieve a $(1.5 + \varepsilon)$-approximation for the problem of cheaply augmenting a given Steiner tree to be a 2-edge-connected Steiner subgraph. This is called the Steiner Tree Augmentation Problem (STAP), and is equivalent to 1-SAG. For the same reason that odd connectivity augmentation reduces to WTAP, this yields a $1.5+\varepsilon$ approximation for $k$-SAG for all odd $k$.

Importantly, this also yields an improved approximation for 1-SCAP. This is because the $k$-SAG and $k$-SCAP problems are equivalent when $k=1$, as any minimal 2-edge-connected Steiner subgraph is also globally 2-edge-connected. However, this unification ceases for $k > 1$, making the higher connectivity setting that we consider much more interesting.

\subsection{Our results}

In this paper, we give the first approximation algorithm with approximation ratio better than 2 for 2-SCAP, and for $k$-SAG for any $k$. To do this we introduce and solve the Steiner Ring Augmentation Problem (SRAP). 

\begin{restatable}[Steiner Ring Augmentation Problem]{problem}{SRAP}\label{prob:srap}
We are given a cycle $H = (V(H), E)$, which is a subgraph of $G = (V,E~ \dot \cup ~L)$. The links $L$ have non-negative costs $c: L \to \mathbb{R}_{\geq 0}$. Furthermore, we are given a set of terminals $R \subseteq V(H)$.

The goal is to select $S \subseteq L$ of minimum cost so that the graph $H' = (V, E \cup S)$ has $3$ pairwise edge-disjoint paths between $u$ and $v$ for all $u,v \in R$.
\end{restatable}

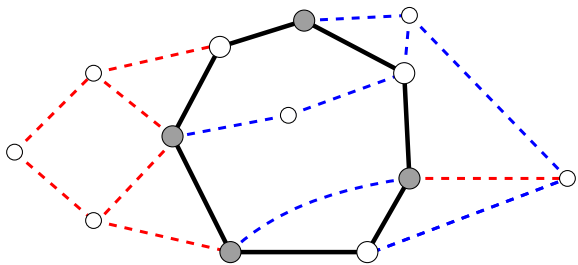
\begin{figure}[h]
		\centering
		\begin{tikzpicture}[scale=0.55]
		
		\begin{scope}
		
		\draw [-] [black,line width=0.6mm,xshift=0 cm] plot [smooth, tension=1] coordinates {(3,6) (1.4,5.5)};

        \draw [-] [black,line width=0.6mm,xshift=0 cm] plot [smooth, tension=1] coordinates {(0.5,3.8) (1.4,5.5)};

        \draw [-] [black,line width=0.6mm,xshift=0 cm] plot [smooth, tension=1] coordinates {(0.5,3.8) (1.6,1.6)};

        \draw [-] [black,line width=0.6mm,xshift=0 cm] plot [smooth, tension=1] coordinates {(4.2,1.6) (1.6,1.6)};

        \draw [-] [black,line width=0.6mm,xshift=0 cm] plot [smooth, tension=1] coordinates {(4.2,1.6) (5,3)};

        \draw [-] [black,line width=0.6mm,xshift=0 cm] plot [smooth, tension=1] coordinates {(5,3) (4.9,5)};

        \draw [-] [black,line width=0.6mm,xshift=0 cm] plot [smooth, tension=1] coordinates {(3,6) (4.9,5)};

		\draw [dashed] [tikzblue, line width=0.4mm,xshift=0 cm] plot [smooth, tension=1] coordinates {(5,3) (3,2.5) (1.6,1.6)};
		\draw [dashed] [tikzblue, line width=0.4mm,xshift=0 cm] plot [smooth, tension=1] coordinates {(3,6) (5,6.1)};
  		\draw [dashed] [tikzblue, line width=0.4mm,xshift=0 cm] plot [smooth, tension=1] coordinates {(5,6.1) (4.9,5)};
           
  		\draw [dashed] [tikzblue, line width=0.4mm,xshift=0 cm] plot [smooth, tension=1] coordinates {(5,6.1) (8,3)};

        \draw [dashed] [red, line width=0.4mm,xshift=0 cm] plot [smooth, tension=1] coordinates {(5,3) (8,3)};

        \draw [dashed] [red, line width=0.4mm,xshift=0 cm] plot [smooth, tension=1] coordinates {(-1,2.1) (0.6,3.8)};

        \draw [dashed] [tikzblue, line width=0.4mm,xshift=0 cm] plot [smooth, tension=1] coordinates {(4.3,1.6) (8,3)};

        \draw [dashed] [tikzblue, line width=0.4mm,xshift=0 cm] plot [smooth, tension=1] coordinates {(4.3,1.6) (8,3)};


        \draw [dashed] [tikzblue, line width=0.4mm,xshift=0 cm] plot [smooth, tension=1] coordinates {(2.7,4.2) (0.6,3.8)};

        \draw [dashed] [tikzblue, line width=0.4mm,xshift=0 cm] plot [smooth, tension=1] coordinates {(2.7,4.2) (4.9,5)};

        \draw [dashed] [red,line width=0.4mm] plot [smooth, tension=1] coordinates {(-2.5,3.5) (-1,2.2)};

        \draw [dashed] [red,line width=0.4mm] plot [smooth, tension=1] coordinates {(-2.5,3.5) (-1,5)};

        \draw [dashed] [red,line width=0.4mm] plot [smooth, tension=1] coordinates {(1.2,5.5) (-1,5)};

        \draw [dashed] [red,line width=0.4mm] plot [smooth, tension=1] coordinates {(0.5,3.8) (-1,5)};

        \draw [dashed] [red,line width=0.4mm] plot [smooth, tension=1] coordinates {(-1,2.2) (1.6,1.6)};

        \draw[black,fill=white] (2.7,4.2) ellipse (0.15 cm  and 0.15 cm);

        \draw[black,fill=white] (-2.5,3.5) ellipse (0.15 cm  and 0.15 cm);

        \draw[black,fill=white] (-1,5) ellipse (0.15 cm  and 0.15 cm);

        \draw[black,fill=white] (-1,2.2) ellipse (0.15 cm  and 0.15 cm);

        \draw[black,fill=white] (5,6.1) ellipse (0.15 cm  and 0.15 cm);

        \draw[black,fill=white] (8,3) ellipse (0.15 cm  and 0.15 cm);

		\draw[black,fill=gray!75,] (3,6) ellipse (0.2 cm  and 0.2 cm);	
		\draw[black,fill=white] (1.4,5.5) ellipse (0.2 cm  and 0.2 cm);			
		\draw[black,fill=gray!75] (0.5,3.8) ellipse (0.2 cm  and 0.2 cm);
		\draw[black,fill=gray!75] (1.6,1.6) ellipse (0.2 cm  and 0.2 cm);	
		\draw[black,fill=white] (4.2,1.6) ellipse (0.2 cm  and 0.2 cm);	
		\draw[black,fill=gray!75] (5,3) ellipse (0.2 cm  and 0.2 cm);
		\draw[black,fill=white] (4.9,5) ellipse (0.2 cm  and 0.2 cm);


		\end{scope}
		\end{tikzpicture}
		
		\caption{A SRAP instance where the black edges denote the given cycle, the dashed edges are the links, and the blue links form a feasible solution. The shaded nodes are the terminals $R$.}
		\label{fig:l2l}
\end{figure}

In terms of approximability, the Steiner Ring Augmentation Problem captures both WCAP and STAP as special cases, in the sense that an $\alpha$-approximation for SRAP implies the same guarantee for both WCAP and STAP. We show that it also implies improved approximation algorithms for 2-SCAP and $k$-SAG.

\begin{restatable}{lemma}{redtoSRAP}\label{lem:redtoSRAP}
    If there is an $\alpha$-approximation for SRAP, then there is an $\alpha$-approximation for $2$-SCAP. If there is an $\alpha$-approximation for SRAP when $R = V(H)$, then there is an $\alpha$-approximation for $k$-SAG.
\end{restatable}

Our main result is an improved approximation algorithm for SRAP.

\begin{restatable}{theorem}{mainTheorem}\label{thm:main_theorem}
    There is a $(1 + \ln{2}  + \varepsilon)$-approximation algorithm for SRAP.
\end{restatable}

Our result on SRAP implies an improved approximation for $k$-SCAP when $k = 2$. 


\begin{corollary}
    There is a $(1 + \ln{2} + \varepsilon)$-approximation algorithm for the 2-Steiner Connectivity Augmentation Problem.
\end{corollary}

This result, together with the results in~\cite{RZZ2022} show that a better-than-2 approximation can be achieved for $k$-SCAP whenever $k \in \{1,2\}$. Recall that in the case of global connectivity augmentation, these two cases are enough to obtain improved algorithms for any $k$. However, these reductions no longer hold in the Steiner setting, where the cuts to be covered are much more complicated. The challenge of obtaining a result for $k$-SCAP when $k \geq 3$ is that the cuts to be covered are no longer a subcollection of the minimum cuts in the given graph. Thus, there is no way to represent these in a cactus or ring structure.

However, in the case of $k$-SAG, we can replace $H$ with a cactus, and subsequently a ring without changing the structure of its minimum cuts. Hence, the $k$-SAG problem ultimately reduces to the special case of SRAP where all cycle nodes are terminals. In this case, we can adapt the local search methodology introduced by Traub and Zenklusen~\cite{TZ2022} to achieve an improved approximation ratio of $(1.5+ \varepsilon)$.

\begin{restatable}{theorem}{approxForSAG}\label{thm:1.5forSAG}
There is a $(1.5+\varepsilon)$-approximation algorithm for SRAP when $R = V(H)$. Hence there is a $(1.5+\varepsilon)$-approximation algorithm for the $k$-SAG problem for any $k$.
\end{restatable}

\subsection{Related Work}
There is a broad body of work on algorithms for network connectivity, for various notions of robustness and in both the weighted and unweighted settings.

We focus on edge-connectivity. Jain's iterative rounding algorithm achieves an approximation ratio of 2 for the general survivable network design problem~\cite{J2001} in the weighted setting. As discussed above, this implies a 2-approximation for minimum cost $k$-edge-connected spanning subgraph (min $k$-ECSS) by taking $r_{ij} = k$ for all pairs of nodes. This is the best currently known approximation ratio for this problem, and can be achieved through a variety of classical methods such as the Primal Dual method, see~\cite{WGMV1993}.

The Weighted Connectivity Augmentation Problem (WCAP) is a special case of min $k$-ECSS where there is a 0 cost $(k-1)$-edge-connected spanning subgraph. For odd $k$, WCAP reduces to the extensively studied Weighted Tree Augmentation Problem (WTAP). Even WTAP is known to be APX-hard and is NP-hard on various special cases such as when the given tree has diameter four~\cite{frederickson1981approximation}\cite{KKL2004}. The current best approximation ratio for WTAP is $(1.5+\varepsilon)$ due to Traub and Zenklusen in~\cite{TZ2022}. Their approach was inspired by the work of Cohen and Nutov~\cite{CN2013} who gave a $(1+\ln{2})$-approximation for WTAP when the given tree has constant diameter. 

It is known that the natural linear programming relaxation for WTAP has integrality gap at most 2 and at least 1.5~\cite{cheriyan2008integrality}. Stronger formulations have been considered for WTAP: Fiorini et.~al. introduced the ODD-LP~\cite{FF2009} which can be optimized over in polynomial time, despite having exponentially many constraints. The ODD-LP has integrality gap at most $2-\frac{1}{2^{h-1}}$ for WTAP on instances where the given tree can be rooted to have height $h$~\cite{PRZ2021}. For unweighted tree augmentation, SDP relaxations with an integrality gap of $(\frac{3}{2} + \varepsilon)$ have been studied~\cite{CG18, CG18a}.

In the case that $k$ is even, WCAP reduces in an approximation preserving way to the Weighted Cactus Augmentation Problem. Traub and Zenklusen broke the barrier of 2 for this problem and gave a $(1.5+\varepsilon)$-approximation algorithm, matching the approximation ratio for WTAP \cite{TZ2022new}.

Improved ratios can be obtained for these and other problems in the unweighted case. For unweighted $2$-ECSS, a factor $\frac{4}{3}$-approximation was known \cite{SV14, HVV19} with recent improvements achieving $1.326$~\cite{GGA23} and finally $(1.3 + \varepsilon)$ \cite{KN23}. For unweighted TAP, the state of the art is a 1.393-approximation due to~\cite{cecchetto2021bridging}, which also holds more generally for the unweighted Connectivity Augmentation Problem. The first better-than-2-approximation for the Forest Augmentation Problem, which generalizes both (unweighted) $2$-ECSS and TAP, was given recently by Grandoni, Jabal Ameli, and Traub~\cite{GAT22}. 

Nutov~\cite{nutov2010nodeweightedSNDP} considered the node-weighted version of SNDP and provided an $O(k\log n)$-approximation for this problem, where $k$ is the maximum connectivity requirement. This implies a tight $O(\log n)$-approximations for the node-weighted versions of all the problems considered in this paper.


\subsection{Preliminaries}\label{sec:prelims}

Suppose $G = (V,E)$ is a graph with vertex set $V$ and edge set $E$. For a non-empty subset of vertices $C \subsetneq V$, the cut $\delta(C)$ consists of all edges in $E$ with exactly one endpoint in $C$. For a subset $X \subseteq E$, we denote by $\delta_X(C) := \delta(C) \cap X$. A cut $C$ is a $k$-cut if $|\delta(C)| = k$.

The edge-connectivity $\lambda(u,v)$ between a pair of vertices $u, v \in V$ is the maximum number of edge-disjoint paths between $u$ and $v$ in $G$. Equivalently, $\lambda(u,v) = \lambda(v,u)$ is the minimum cardinality of a cut $\delta(C)$ with $u \in C$ and $v \not \in C$. A graph is said to be $k$-edge-connected if $\lambda(u,v) \geq k$ for all pairs $u,v \in V$. Given a subset of terminals $R \subseteq V$, we say that $G$ is Steiner $k$-edge-connected on $R$ if $\lambda(u,v) \geq k$ for all pairs of terminals $u,v \in R$.

Given a Steiner $k$-edge-connected graph $G = (V,E)$ on terminals $R$, fix $r \in R$, and define $$\mathcal{C}'' := \{C \subseteq V \setminus r: |\delta(C)| = k, C \cap R \neq \emptyset\}$$ to be the family of $k$-cuts of $G$ which separate some terminal from $r$. We call the cuts in $\mathcal{C}''$ \textbf{dangerous cuts}.

The $k$-SCAP problem is a hitting set problem where the ground set is the collection of links $L$, and the sets are $\delta(C)$ where $C \in \mathcal{C}''$. That is, a set of links which is a solution to this hitting set problem will cause the graph to become Steiner $(k+1)$-edge-connected. The link $\ell$ ``covers'' those dangerous cuts $C$ for which $\ell \in \delta(C)$.

Recall the Steiner Ring Augmentation Problem.

\SRAP*

We refer to the nodes in $R$ as $\textbf{terminals}$, and the nodes in $V \setminus R$ as \textbf{Steiner nodes}. We will also refer to $H$ as the \textbf{ring} to distinguish it from a generic cycle. It will be convenient to fix a root $r \in R$ of the ring and an edge $e_r \in E$ incident on $r$. 

We now introduce some terminology which allows us to phrase the SRAP problem as a covering problem on a collection of ring-cuts only. Indeed, given the ring $H = (V(H), E)$, denote the set of min-cuts of $H$ as:
$$\mathcal{C}' = \{C \subseteq V(H) \setminus r: |\delta_E(C)| = 2\}.$$

Since we are only interested in connectivity between the terminals, we only need to cover the subfamily of $\mathcal{C}'$ which separates terminals. Let 
$$\mathcal{C} = \{C \subseteq V(H) \setminus r : |\delta_E(C)| = 2, C \cap R \neq \emptyset\}.$$
We call the cuts in $\mathcal{C}$ \textbf{dangerous ring-cuts}.

We will use a similar notion of ``full components" as introduced in Ravi, Zhang and Zlatin~\cite{RZZ2022} for STAP. Consider any solution $S \subseteq L$ to SRAP.

\begin{definition}
A \textbf{full component} of a SRAP solution $S$, is a maximal subtree of the solution where each leaf is a ring node (that is, a vertex of $V(H)$), and each internal node is in $V \setminus V(H)$.
\end{definition}

It is a basic fact that any link-minimal SRAP solution can be uniquely decomposed into link-disjoint full components. 

\begin{definition}
Let $S$ be a solution to SRAP. We say that a set $A \subseteq V(H)$ is \textbf{joined} by $S$ if there is a full component with leaves $A$.
\end{definition}

\begin{definition}
We say that a cut $C \in \mathcal{C}$ is \textbf{covered} by a solution $S$ if $A \cap C \neq \emptyset$ and $ A \cap \bar C \neq \emptyset$ for some subset of ring nodes $A$ which are joined by a full component of $S$. \end{definition}

Hence, we can think of the SRAP problem as the problem of hitting the dangerous ring-cuts with full components.

\begin{lemma}\label{lem:feasible-iff-covered}
A solution $S$ is feasible for SRAP iff all dangerous ring-cuts are covered by $S$.
\end{lemma}
\begin{proof}
    Suppose that each dangerous ring-cut of $H = (V(H),E)$ is covered. That is, for each $C\in \mathcal{C}$, there is a full component of $S$ joining vertices $A \subseteq V(H)$ such that $A \cap C \neq \emptyset$ and $A \cap \bar C \neq \emptyset$. We will show that any cut $C' \subseteq V$ which separates a pair of terminals has $|\delta_{E \cup S}(C')| \geq 3$, thereby implying the existence of 3 pairwise edge-disjoint paths between all terminal pairs. 
    
    Consider the cut $C' \cap V(H)$ obtained by restricting $C'$ to $V(H)$. We may assume that $C' \cap V(H)$ is a 2-cut, otherwise $|\delta_E(C')| \geq 3$ already. Furthermore, note that $C' \cap R \neq \emptyset$. Hence, $C' \cap V(H)$ is a dangerous ring-cut, and there is some full component of $S$ covering it. This implies that $|\delta_S(C')| \geq 1$ so $|\delta_{E \cup S}(C')| \geq 3$ as desired.

    For the reverse direction, suppose that $S$ is a feasible solution to SRAP. Consider any dangerous ring-cut $C \in \mathcal{C}$. We will show that there must be a path in $S$ from a node in $C$ to another ring node which is not contained in $C$, implying that there is some full component covering $C$. Since $|\delta_{E \cup S}(C)| \geq 3$, there must be a link $(u,v) \in S$ with $u \in C$ and $v \not \in C$. If $v$ is a ring node, we are done. Otherwise, the cut $C' = C \cup v$ is another dangerous cut which is covered by $S$. Repeating this process eventually yields a path of links from a node in $C$ to a ring node outside $C$, hence the full component of $S$ containing this path covers $C$.
\end{proof}

This motivates the definition of the Hyper-SRAP problem: we are given a ring $H = (V(H),E)$ with terminals $R \subseteq V(H)$, a root vertex $r \in R$, and a collection of hyper-links $\mathcal{L} \subseteq 2^{V(H)}$ with non-negative costs $c: \mathcal{L} \to \mathbb{R}_{\geq 0}$.

Let $\mathcal{C} = \{C \subseteq V(H) \setminus r: |\delta_E(C)| = 2, C \cap R \neq \emptyset \}$ be the set of dangerous ring-cuts. A cut $C \in \mathcal{C}$ is \textbf{covered} by a hyper-link $\ell$ if $\ell \cap C \neq \emptyset$ and $\ell \cap \bar C \neq \emptyset$. The Hyper-SRAP problem is to find a minimum cost subset of hyper-links so that all cuts in $\mathcal{C}$ are covered.

We will use the notion of the hyper-link intersection graph to characterize feasible solutions to Hyper-SRAP. First, we define what it means for two hyper-links to be intersecting.



\begin{definition}
    Let $\ell$ and $\ell'$ be a pair of hyper-links. Let $(v_1, \ldots, v_k)$ be the sequence of vertices of $\ell \cup \ell'$ obtained by traversing the ring (in either direction). Then $\ell$ and $\ell'$ are \textbf{intersecting} if there are vertices $v_{i_1},v_{i_2},v_{i_3},v_{i_4}$ with $i_1 < i_2 < i_3 < i_4$ such that $v_{i_1}, v_{i_3} \in \ell$ and $v_{i_2}, v_{i_4} \in \ell'$. 
\end{definition}

Given an instance of Hyper-SRAP with ring $H = (V(H),E)$ and hyper-links $\mathcal{L}$, we define the hyper-link intersection graph $\Gamma$ as follows. For each hyper-link $\ell \in \mathcal{L}$ there is a node $v_\ell$. Two nodes $v_{\ell_1}$ and $v_{\ell_2}$ are adjacent in the hyper-link intersection graph if and only if $\ell_1$ and $\ell_2$ are intersecting hyper-links. For ring vertices $u,v \in V(H)$, we say that there is a path from $u$ to $v$ in $\Gamma$ if there is a path in $\Gamma$ from a hyper-link containing $u$ to a hyper-link containing $v$.

\begin{definition}
We say that a full component is \textbf{$\gamma$-restricted} if it joins at most $\gamma$ ring nodes. We say that a solution to SRAP is \textbf{$\gamma$-restricted} if it uses only $\gamma$-restricted full components. Analogously, we say an instance of Hyper-SRAP is \textbf{$\gamma$-restricted} if each hyper-link has size at most $\gamma$. 
\end{definition}

Similar to the approach developed in~\cite{RZZ2022} for STAP, we can work with $\gamma$-restricted solutions to SRAP while losing an arbitrarily small constant in the approximation ratio.
This follows from a result of Borchers and Du for Steiner trees~\cite{BD1997}.

\begin{lemma}
\label{lem:gamma_restriction}
For an instance of SRAP, let $S^* \subseteq L$ be an optimal solution and $S_\gamma \subseteq L$ be an optimal $\gamma$-restricted solution, where $\gamma(\varepsilon)= 2^{\lceil{{\frac{1}{\varepsilon}}\rceil}}$ for some $\varepsilon > 0$. Then $\frac{c(S_\gamma)}{c(S^*)} \leq 1 + \varepsilon.$
\end{lemma}
\begin{proof}
By the optimality of $S^*$, each full component of $S^*$ is minimum cost Steiner tree on the ring nodes $A$ that it joins, in the graph $(A \cup (V \setminus V(H)),L)$. By \cite{BD1997}, there is a $\gamma$-restricted Steiner Tree joining $A \subseteq V(H)$ in this graph with at most $(1 + \varepsilon)$ times the cost. Applying this to each full component of $S^*$ yields a solution of cost at most $(1+\varepsilon)c(S^*)$ which only has full components joining at most $\gamma$ ring nodes.
\end{proof}

Recall that the minimum cost Steiner tree problem can be solved in polynomial time when the number of terminals is constant.

\begin{theorem}[Dreyfus and Wagner~\cite{DW1971}]\label{thm:Steiner-tree-FPT}
    The minimum Steiner tree problem can be solved in time $O(n^3 \cdot 3^p)$ where p is the number of terminals.
\end{theorem}

Because \Cref{lem:feasible-iff-covered} shows that the feasibility of a solution only depends on the nodes that are joined by full components of $S$, we can effectively disregard the Steiner nodes outside the ring and observe that any instance of SRAP is equivalent to an instance of Hyper-SRAP in which the cost of a hyper-link $A$ is the minimum cost Steiner tree connecting $A$ in the graph $(A \cup (V \setminus V(H)), L)$. By \Cref{lem:gamma_restriction}, and \Cref{thm:Steiner-tree-FPT}, we can perform this reduction from an arbitrary SRAP instance to an instance of $\gamma$-restricted Hyper-SRAP in polynomial time while only losing a factor of $(1+\varepsilon)$ in the approximation ratio.  

Finally, we will make use of directed solutions to the SRAP problem. If $\vecF$ is a collection of directed links between pairs of vertices of the ring, then a dangerous ring-cut $C \in \mathcal{C}$ is covered by $\vecF$ if $\delta_{\vecF}^-(C) \neq \emptyset$, i.e. if there is an arc in $\vecF$ which enters $C$. Then $\vecF$ is a feasible \textbf{directed solution} if all dangerous ring-cuts are covered by $\vecF$. Analogously, if $S$ is a set of undirected links and $\vecF$ is a set of directed links then we will say that $S \cup \vecF$ is a feasible \textbf{mixed solution} if every dangerous ring-cut is covered by $S$ or by $\vecF$. 

\section{Overview of Techniques}
The main contribution of this article is to prove \Cref{thm:main_theorem}.

\mainTheorem*

We show this implies improved approximation guarantees for both 2-SCAP and $k$-SAG. Recall that in the $k$-SCAP problem, we are given a graph $H$ which is Steiner $k$-edge-connected on the terminal set $R$. We want to augment this graph by including additional links of minimum cost so that there exists $k+1$ pairwise edge-disjoint paths between every pair of terminals. This is equivalent to covering the dangerous cuts of $H$ with a minimum cost set of links.

Note that if $k$ is larger than the global edge-connectivity of $H$, then there may be exponentially many dangerous cuts. Indeed even when $|R| = 2$, the number of minimum $\{s,t\}$-cuts in a graph $G$ on $n$ nodes may be exponential in $n$. This shows that, unlike in global connectivity augmentation, the set of cuts to be covered in the $k$-SCAP problem cannot be represented by a cactus structure that is efficiently computable. Although it is true that there are polynomially many dangerous cuts up to separating the same subset of terminals~\cite{dinitz1994carcass}, this is not sufficient for our purposes.

However, in the case of $k = 2$, we show that we may take $H$ to be a globally 2-edge-connected graph rather than merely Steiner 2-edge-connected. This allows us to compactly represent the cuts to be covered by a cactus on a set of nodes, some of which are terminals. We then use the technique introduced by G\a'alvez et.~al.~\cite{galvez2021cycle} to replace the cactus with a ring by adding in links of 0 cost. This brings the problem into the SRAP framework. In the case of $k$-SAG, we are guaranteed that $H$ is $k$-edge-connected so we can directly replace $H$ with a cactus with the same cut structure. This yields:

\redtoSRAP*

\subsection{The Relative Greedy Algorithm for WRAP}
We begin by briefly describing the relative greedy framework which is used by Traub and Zenklusen in~\cite{TZ2022new} for the standard Weighted Ring Augmentation Problem. 
Traub and Zenklsuen first introduce a directed weakening of the WRAP problem which allows them to obtain an initial directed 2-approximate solution with good structure. At a high level, each undirected link $\ell = (u,v)$ in the instance is replaced by a pair of directed links $(u,v)$ and $(v,u)$ of the same cost. These directed links are called ``shadows" of $\ell$ and cover a subset of the cuts which $\ell$ does. This directed weakening of WCAP can be solved optimally in polynomial time, yielding a 2-approximate directed solution for WCAP.

This solution is then improved over time by iteratively adding collections of links, while dropping links from the initial solution so that feasibility is maintained. Crucially, 
the link subsets which are added in each step are restricted to be so-called $``\alpha$-thin" subsets for some constant $\alpha \in \mathbb{N}$. Traub and Zenklusen show that the ``best'' $\alpha$-thin subset of links can be found efficiently in each step (their Optimization Theorem), while also being a broad enough class so that the absence of an improving $\alpha$-thin subset implies good guarantees on the cost of the current solution (the Decomposition Theorem). At a high level, these three ingredients: a structured 2-approximation, an optimization theorem, and a decomposition theorem, are enough to implement the relative greedy algorithm yielding an approximation guarantee of $1 + \ln(2) + \varepsilon$. 



\subsection{The Relative Greedy Algorithm for SRAP}
To employ a relative greedy algorithm for SRAP, we show how each of these three ingredients can be enacted in the Steiner setting. The first challenge is to obtain an initial 2-approximate directed solution on the nodes of the ring. While a 2-approximation for SRAP follows from e.g.~Jain's result, this solution may use links between Steiner vertices outside the ring, and in general does not have enough structure to be useful for the relative greedy framework.

Hence, we begin by showing how to compute a directed 2-approximation with helpful structural properties. To do this, we must perform several preprocessing steps on the given SRAP instance to ensure that it is both ``shadow-complete" and ``metric-complete" simultaneously. We call such instances ``complete'' instances and show that a complete instance can be obtained without loss of generality in at most three rounds of metric and shadow completion.

Next, we compute an initial, highly structured, directed 2-approximation for SRAP, which we call an $R$-special solution. Roughly speaking, an $R$-special solution is a directed solution which is a planar $r$-out-arborescence, has out-degree at most 2, and is only incident on the terminals $R$. We more formally define the properties that are required of this initial solution in \Cref{sec:rspecial}, and describe how to compute one in polynomial time on any complete SRAP instance.

We can now iteratively improve upon our initial 2-approximation by adding and dropping links while maintaining feasibility. However, unlike in standard WRAP, the subsets of links we add correspond to $\alpha$-thin subsets of \textit{hyper-links}, rather than standard links of size 2. We utilize \Cref{lem:gamma_restriction} and \Cref{thm:Steiner-tree-FPT} to compute the costs of all hyper-links of constant size at most $\gamma$. Then, in each iteration, we greedily select an $\alpha$-thin subset of hyper-links to add to our solution, dropping links from the initial directed solution while maintaining feasibility. The collection is chosen so as to minimize the ratio of costs between the hyper-links added and the directed links which can be dropped as a result of adding these hyper-links. 
At the end of the algorithm, we include all links in each full component corresponding to chosen hyper-links.

In \Cref{sec:decompandopt}, we describe how we can prove an optimization theorem so that each step can be executed in polynomial time. Also in~\Cref{sec:decompandopt} we discuss the main decomposition theorem for Hyper-SRAP, which is necessary for proving the desired approximation guarantee. 

The full details of the relative greedy algorithm are described in \Cref{app:sec:relgreedyalgo}, where we prove \Cref{thm:main_theorem}.

\subsection{An $R$-special 2-approximate solution}\label{sec:rspecial}


In \Cref{app:sec:2approx}, we show that for any complete instance of SRAP, we can compute a directed 2-approximate arborescence solution which is only incident on terminals. We call a directed solution which satisfies the properties in~\Cref{thm:2approx} an \textbf{$R$-special} solution. Let $\OPT$ denote the cost of the optimal solution.
\begin{restatable}{theorem}{twoApprox}\label{thm:2approx}
    Given a complete instance of SRAP, there is a polynomial time algorithm which yields a directed solution $\vecF$ of cost at most $2\OPT$ such that: 

    \begin{enumerate}
        \item $\vecF$ is only incident on the terminals $R$
        \item $(R,\vecF)$ is an $r$-out arborescence.
        \item $(R, \vecF)$ is planar when $V(H)$ is embedded as a circle in the plane.
        \item For any $v \in V$, no two directed links in $\delta^+_{\vecF}(v)$ go in the same direction along the ring.
    \end{enumerate}
\end{restatable}

The proof of \Cref{thm:2approx} is significantly more involved than the directed 2-approximation for standard WRAP. Unlike in WRAP, the optimal solution for SRAP may use links between Steiner nodes outside the ring in order to increase the connectivity between pairs of terminals in the ring. Nonetheless, using analagous techniques to~\cite{RZZ2022} for STAP, we can obtain a 2-approximate solution consisting of directed links only between ring nodes. Essentially, taking a directed cycle on the nodes joined by each full component of the optimal solution yields a directed solution of at most 2 times the cost. We can do the same for SRAP, yielding:

\begin{lemma}\label{lem:2appx_cycles}
    Given a complete instance of SRAP, there is a directed solution of cost at most $2\OPT$ consisting of a collection of directed cycles on the ring nodes. 
\end{lemma}

At this point we could iteratively shorten the solution to obtain a 2-approximate planar $r$-out arborescence as in~\cite{TZ2022new}. However, this is not sufficient for our purposes as this directed solution still uses Steiner nodes \textit{inside} the ring. 

Hence, to prove \Cref{thm:2approx}, we will begin with the 2-approximate directed solution consisting of a collection of directed cycles on the ring nodes as guaranteed by \Cref{lem:2appx_cycles}. Then, we prove a ``cycle merging lemma" which allows us to iteratively merge these cycles to eventually obtain a feasible directed solution which is a \textit{single} directed cycle of at most the cost. This yields: 

\begin{lemma}\label{lem:single-cycle}
    Given a complete instance of SRAP, let the optimal solution have cost $\OPT$. Then there exists a directed solution of cost at most $2\OPT$ whose links form a single directed cycle containing $R$.
\end{lemma}

Finally, we can use the links added in the second metric completion step to shortcut over the Steiner nodes in the ring to obtain a directed cycle solution that only touches terminals. 

\begin{lemma}\label{lem:terminal-cycle}
    Given an instance of SRAP, if the optimal solution has cost $\OPT$, then there is a directed solution of cost at most $2\OPT$ whose links consist of a single directed cycle with node set $R$.
\end{lemma}


Having shown that there exists a directed 2-approximate solution to any complete SRAP instance which is incident only on the terminals $R$, we can now shorten the directed links in this solution to obtain the desired arborescence structure. This follows directly from the planar and structural properties that Traub and Zenklusen show hold for any non-shortenable directed WRAP solution~\cite[Theorem 2.6]{TZ2022new}, and our \Cref{thm:2approx} follows.


\subsection{Decomposition and Optimization Theorems}\label{sec:decompandopt}

An $R$-special directed solution is necessary because it allows us to leverage a decomposition result on directed solutions with respect to the optimum. Because of the decomposition result, it can be argued that every iteration of the algorithm will find an improving local move as long as the current solution is expensive. Traub and Zenklusen prove a decomposition result of this kind to bound the approximation ratio of the relative greedy algorithm for WRAP~\cite{TZ2022new}. In our setting, we need an analagous decomposition theorem for the hyper-links which correspond to the full components in the optimal SRAP solution. We prove the decomposition theorem for arbitrary Hyper-SRAP instances with respect to any $R$-special directed solution. 

First we need to define the notion of an $\alpha$-thin collection of hyper-links. Recall that $\mathcal{C}'$ is the set of minimum cuts of the ring $H$ which do not include the root.

\begin{definition}
    A collection of hyper-links $K \subseteq \mathcal{L}$ is $\alpha$-thin if there exists a maximal laminar subfamily $\mathcal{D}$ of $\mathcal{C}'$ such that for each $C \in \mathcal{D}$, the number of hyper-links in $K$ which cover $C$ is at most $\alpha$.
\end{definition} 


With this definition, we can state \Cref{thm:decomposition} for Hyper-SRAP.

\begin{restatable}[Decomposition Theorem]{theorem}{decompositionTheorem}\label{thm:decomposition}
Given an instance of Hyper-SRAP $(H = (V(H),E), R, \mathcal{L})$, suppose $\vecF_0$ is an $R$-special directed solution and $S \subseteq \mathcal{L}$ is any solution. Then for any $\varepsilon > 0$, there exists a partition $\mathcal{Z}$ of $S$ into parts so that:
\begin{itemize}
    \item For each $Z \in \mathcal{Z}$, $Z$ is $\alpha$-thin for $\alpha = 4\lceil{1/\varepsilon}\rceil$.
    \item There exists $Q \subseteq \vecF_0$ with $c(Q) \leq \varepsilon \cdot c(\vecF_0)$, such that for all $f \in \vecF_0 \setminus Q$, there is some $Z \in \mathcal{Z}$ with $f \in \text{drop}_{\vecF_0}(Z)$. That is, $\vecF_0 \setminus Q \subseteq \bigcup_{Z \in \mathcal{Z} } \drop_{\vecF_0}(Z)$.
\end{itemize}

\end{restatable}

In the above theorem, for a collection of hyper-links $K$ and an $R$-special directed solution $\vecF_0$, the notation $\drop_{\vecF_0}(K)$ denotes a set of directed links from $\vecF_0$ which can be dropped while preserving feasibility of $\vecF_0 \cup K$. For each directed link $f \in \vecF_0$, we describe a collection of dangerous ring-cuts for which $f$ is ``responsible", denoted $\mathcal{R}(f)$, and defined formally in \Cref{app:sec:dropping_dir}. 

Then $$\drop_{\vecF_0}(K) := \{f \in \vecF_0 : |\delta_{K}(C)| \geq 1 \text{ for all } C \in \mathcal{R}_{\vecF_0}(f)\}$$ is the set of all directed arcs such that the cuts they are responsible for are covered by $K$. 

This definition was used by Traub and Zenklusen~\cite{TZ2022new} for WRAP. They show that a link $(u,v)$ can be dropped if and only if $v$ is connected to a ``$v$-good'' vertex in the link-intersection graph restricted to $K$. In~\cite{TZ2022new}, a $v$-good vertex is a vertex which is not a descendant of $v$ with respect to the initial 2-approximate arborescence.

This is no longer the case in our setting, since it is possible to drop $(u,v)$ even if $K$ connects $v$ to some Steiner node in the ring. Hence, we redefine the notions of $v$-good and $v$-bad to characterize when a directed link can be dropped in our setting.
We prove the following.

\begin{lemma}\label{lem: whencanIdrop}
For a collection of hyper-links $K$, a directed link $(u,v)$ is in $\drop_{\vecF_0}(K)$ if and only if $\Gamma(K)$ contains a path from a hyper-link containing $v$ to a hyper-link containing a $v$-good vertex $w$. 
\end{lemma}

Note that \Cref{lem: whencanIdrop} does not necessarily hold if $\vecF_0$ is not $R$-special, since a directed link entering a Steiner node could be droppable even if $K$ has no hyper-links containing it. We also show that an $R$-special solution $\vecF_0$ can be augmented with artificial links to obtain a $V(H)$-special solution $\vecF'$ so that the set of $v$-bad nodes in $\vecF_0$ correspond to the set of descendants of $v$ in the arborescence $\vecF'$. 




We will use our characterization of when a directed arc is droppable to prove \Cref{thm:decomposition}. We follow the approach in~\cite{TZ2022new} which proves the result when all hyper-links have size 2 and $R = V(H)$. We construct a ``dependency graph" which allows us to partition the links of $S$ into the desired $\alpha$-thin pieces. In our hyper-link setting, the nodes of the dependency graph correspond to ``festoons" composed of hyper-links rather than festoons of links of size 2.



\Cref{lem:gamma_restriction} and \Cref{thm:Steiner-tree-FPT} allow us to convert a given instance of SRAP into an equivalent Hyper-SRAP instance efficiently. It also ensures that each local move of the algorithm runs in polynomial time. In each local move, the algorithm chooses amongst all $\alpha$-thin collections of hyper-links $K$, the choice which minimizes the ratio between the cost of the hyper-links in $K$, and the cost of the directed links which will be dropped as a result of adding $K$ to the solution. We show that this operation can be performed in polynomial time as long as all hyper-links have size at most $\gamma$.  

\begin{theorem}\label{thm:optimization}
    Given an instance of $\gamma$-restricted Hyper-SRAP, an $R$-special directed solution $\vecF_0$, a set of directed links $\vecF \subseteq \vecF_0$, and an integer $\alpha \geq 1$, there is a polynomial time algorithm which finds a collection of hyper-links $K$ minimizing $$\frac{c(K)}{c(\drop_{\vecF_0}(K)\cap \vecF)}$$ over all $\alpha$-thin subsets of hyper-links.
\end{theorem}

\Cref{thm:2approx,thm:decomposition,thm:optimization} together are essentially enough to employ a relative greedy strategy to obtain a $(1+\ln{2} + \varepsilon)$-approximation for SRAP. Beginning with an $R$-special 2-approximate directed solution, we iteratively add to this solution a greedily chosen $\alpha$-thin collection of hyper-links in each step. When a particular collection of hyper-links $K$ is chosen to be added, the directed links in $\drop_{\vecF_0}(K)$ are dropped, and this process is repeated until $\vecF_0$ becomes empty. 


\subsection{A Local Search Improvement for $k$-SAG}
Additionally, in the case that all ring nodes are terminals, we can use a local search framework introduced in~\cite{TZ2022} to give an algorithm with an improved approximation guarantee of $(1.5+\varepsilon)$. Together with \Cref{lem:redtoSRAP}, this yields:

\approxForSAG*

The key idea in the improvement is to not only consider dropping links from the initial directed solution, but to drop undirected links which were added in previous iterations by associating to each undirected link a witness set of directed links which indicate when it can be dropped.

We are able to do this in the case that $V(H) = R$ because in this case, \textit{any} feasible directed solution can be transformed into an $R$-special solution of at most the cost. However, this is not the case for general SRAP.



This illustrates why it may be surprising that \Cref{thm:2approx}, our key technical contribution, holds. While it is not true that any feasible directed solution can be made $R$-special, the proof of \Cref{thm:2approx} shows that a directed solution which consists of a collection of directed cycles on the ring \textit{can} be made $R$-special, and this is enough to prove the desired guarantees with respect to the optimum.

\section{Reductions to SRAP}\label{app:sec:srap-reductions}
In this section, we prove \Cref{lem:redtoSRAP}.

\redtoSRAP*


Consider an instance of the $2$-SCAP problem, where we are given a graph $H$ which is Steiner 2-edge-connected for a set of terminals $R$, and recall that $\mathcal{C}'' = \{C \subseteq V \setminus r: |\delta(C)| = 2, C \cap R \neq \emptyset\}$ is the set of dangerous cuts to be covered.

We will show that the connected component of $H$ which contains $R$ can be assumed to be genuinely 2-edge-connected, rather than merely Steiner 2-edge-connected.

\begin{lemma}
    We may assume without loss of generality that the connected components of $H$ consist only of isolated Steiner nodes, and a single 2-edge-connected loopless component $H'$ containing all of the terminals $R$.
\end{lemma}
\begin{proof}
    First, since $H$ contains a path between every pair of terminals, there must be a single connected component containing $R$, call it $H'$. So all other connected components contain only Steiner nodes.

    Now consider a connected component $K$ of Steiner nodes. Observe that any dangerous cut $C$ of size $\abs{\delta_E(C)} = 2$ separating terminals must either contain $K$, or be disjoint from $K$. Otherwise, the cut $C \cup K$ is still dangerous, but $\abs{\delta_E(C \cup K)} < 2$, contradicting the Steiner 2-edge-connectivity of $H$. Hence, $K$ may be contracted into a single isolated Steiner node. 

    Now, if $H'$ contains any cut $K$ of size $\abs{\delta_E(K)} = 1$, it cannot separate terminals, so assume $K \cap R = \varnothing$. Denote the edge crossing this cut as $\delta_E(K) = \{(u, v)\}$ with $u \in K$, $v \not \in K$. Similarly to above, we observe that any dangerous cut $C$ containing $v$ must contain $K$. Otherwise, the cut $C \cup K$ is still dangerous, but $\abs{\delta_E(C \cup K)} < 2$. So, again, we may contract $K \cup v$ without changing the structure of cuts to be covered. 

    Such a contraction removes at least one cut of size 1, so we may repeat this iteratively until no cuts of size 1 remain. This leaves us with $H'$ being a 2-edge-connected graph (and observe that neither of the contraction operations we performed could have introduced a loop). 
\end{proof}

With the above lemma, it is easy to see that any dangerous cut $C$ of $H$ is obtained by taking a dangerous cut of $H'$ (that is, a 2-cut $C \subseteq V(H')$ which separates terminals), and adding in some Steiner nodes in $V \setminus V(H')$. At a high level, this allows us to replace $H'$ with a different graph with the same cut structure without meaningfully changing the problem.

We can now apply an easy extension of a theorem of Dinits, Karzanov, and Lomonosov on the cactus representation of the min-cuts of a graph.

\begin{theorem}[Cactus representation of min-cuts, \cite{DKL1976}]\label{app:thm:cactus-rep}
Let $G = (V, E)$ be a loopless graph. There is a cactus $\widehat{G} = (U, F)$ and a map $\phi: V \to U$ such that for every 2-cut of $\widehat{G}$ with shores $U_1$ and $U \setminus U_1$, the preimages $\phi^{-1}(U_1)$ and $\phi^{-1}(U \setminus U_1)$ are the two shores of a min-cut in $G$. Moreover, every min-cut of $G$ arises in this way. 
\end{theorem}

The following strengthening of \Cref{app:thm:cactus-rep} immediately follows by letting $u \in U$ be in $R'$ if and only if $\phi^{-1}(u)$ contains a node of $R$.

\begin{corollary}\label{app:cor:terminal_cactus}
    Let $G = (V,E)$ be a loopless graph with terminals $R \subseteq V$. There is a cactus $\widehat{G} = (U,F)$, a set $R' \subseteq U$, and a map $\phi: V \to U$ such that for every 2-cut of $\widehat{G}$ separating nodes of $R'$, with shores $U_1$ and $U\setminus U_1$, the preimages $\phi^{-1}(U_1)$ and $\phi^{-1}(U \setminus U_1)$ are the two shores of a min-cut in $G$ which separates terminals. Also, every min-cut in $G$ separating terminals arises in this way. 
\end{corollary}

Applying \Cref{app:cor:terminal_cactus} to $H'$, we may replace $H'$ by a cactus $\widehat{H'}$ such that there is a correspondence between 2-cuts which separate terminals in $H'$ and 2-cuts in which separate nodes of $R'$ in $\widehat{H'}$. All links in $G$ incident to a node $v$ of $H'$ will now be incident to the corresponding node $\phi(v)$ in $\widehat{H'}$. Thus, there is a correspondence between feasible solutions to 2-SCAP and the problem of 2-SCAP where the connected component of $H$ containing $R$ is a cactus. In particular, we have shown that 2-SCAP reduces to the following problem:

\begin{problem}[Steiner Augmentation of a Cactus]
We are given a cactus $H = (V(H), E)$, which is a subgraph of $G = (V,E \dot \cup L)$. The links $L$ have non-negative costs $c: L \to \mathbb{R}_{\geq 0}$. There is also a set of terminals $R \subseteq V(H)$. 

The goal is to select $S \subseteq L$ of cheapest cost so that the graph $H' = (V, E \cup S)$ has $3$ pairwise edge-disjoint paths between $u$ and $v$ for all $u,v \in R$.
\end{problem}

By the addition of zero-cost links (c.f Theorem 4 in~\cite{galvez2021cycle}), we can unfold the cactus further so that $H$ is a cycle. This is precisely the Steiner Ring Augmentation Problem, which we restate here:

\SRAP*

Each of these reductions can be performed in polynomial time. And, since the structure of dangerous cuts to be covered remains the same, any solution to the $2$-SCAP instance gives rise a solution of the same cost in the SRAP instance, and every solution in the SRAP instance arises in such a way. This proves the first half of \Cref{lem:redtoSRAP}.

The second half of \Cref{lem:redtoSRAP}, the reduction from $k$-SAG, follows in an almost identical (albeit simpler) fashion. We replace the $k$-edge-connected graph $H$ with a cactus $\widehat{H}$ using \Cref{app:thm:cactus-rep}, such that there is a correspondence between the min-cuts of $H$ and the 2-cuts of $\widehat{H}$. This time, all nodes in $\widehat{H}$ are terminals. We can again add zero cost links to replace $\widehat{H}$ with a cycle, yielding an instance of SRAP in which $R = V(H)$. 

\section{Complete Instances}\label{app:sec:complete-instances}
In this section, we define a notion of a \textbf{complete} SRAP instance and show that we may assume that we are working with instances of this form without loss of generality. This ensures that the necessary links are available in order to prove \Cref{thm:2approx}.

\begin{definition}
    An instance of SRAP is called a \textbf{complete instance} if for every $u, v \in V(H)$, there is a directed link $(u, v)$ whose cost is equal to the shortest directed path from $u$ to $v$, and for every directed link $(u, v)$, the instance contains all shortenings of $(u, v)$ with at most cost $c((u, v))$. 
\end{definition}
 
To obtain a complete instance, we will iteratively perform two types of completion steps on our instance: metric completion and shadow completion.

For the first metric completion, we place a new undirected link between each pair of vertices of the ring $u,v \in V(H)$, with cost equal to the shortest path from $u$ to $v$ which uses only links (if there is no link-path from $u$ to $v$ the cost is infinity). This does not affect the cost of the optimal solution, so we may assume without loss of generality that the instance is metric complete. After this step, there is an undirected link between every pair of vertices in the ring.

For shadow completion, for each link $\ell = (u,v)$ with $u,v \in V(H)$, we will add to the instance a collection of directed links known as the \textbf{shadows} of $\ell$. These will all have the same cost as $\ell$. The shadows of $\ell$ consist of the two directed links $(u,v)$ and $(v,u)$ as well as all shortenings of these two directed links. If $(u,v)$ is a directed link, then $(s,t)$ is a \textbf{shortening} of $(u,v)$ if $v = t$ and $s$ is a vertex on the path from $u$ to $v$ in $E \setminus e_r$. With this definition, it is not hard to see that a shadow of $\ell$ covers a subset of the cuts in $\mathcal{C}$ that $\ell$ covers. Hence, we may perform this step without affecting the cost of the optimal solution.

Finally, we do a second round of metric completion on the newly added directed links, similar to the first. For every pair of vertices $u,v \in V(H)$, if there is a directed path from $u$ to $v$, we add a directed link $(u,v)$ with cost equal to the cost of the shortest such path. 

After these operations, we are left with an instance of SRAP such that any solution to this instance can be converted into a solution to the original SRAP instance with the same cost. Hence, we may assume that these operations have been performed on the given SRAP instance without loss of generality. See \Cref{app:fig:complete-instance} for an example of this preprocessing.

Furthermore, we claim that any further iterations of these two preprocessing operations (metric completion and shadow completion) will not change the instance. 

\begin{restatable}{lemma}{completeinstance}\label{lem:complete-instance}
Any SRAP instance can be made complete by performing metric completion, then shadow completion, then a second metric completion. 
\end{restatable}

\begin{proof}
    Let $L$ denote the set of links in the original instance, and $L^1, L^2, L^3$ denote the links after each of the three preprocessing steps, respectively. We want to show that $L^3$ is a complete instance. Clearly the instance is already metrically complete, since the final preprocessing operation is a metric completion. So it suffices to show that all shortenings of links in $L^3$ already exist in $L^3$. 

    Consider any link $(u, v) \in L^3$. If $(u, v)$ did not arise from the second metric completion step, then it is in $L^2$ and all of its shortenings were added in the shadow completion step. 

    So we focus on the case in which $(u, v)$ arose as a metric completion of some path of links $(u=w_0, w_1), (w_1, w_2), \ldots, (w_{k-1}, w_k = v)$, where all links on this path are in $L^2$. Without loss of generality, suppose that $u$ lies to the left of $v$. Consider an arbitrary shortening of $(u,v)$, which takes the form $(s, v)$ for some $s$ lying between $u$ and $v$. There must be some link $(w_i, w_{i + 1})$ on the path such that $w_i$ lies to the left of $s$, and $w_{i + 1}$ to the right of $s$. Then $(s, w_{i+1})$ is a shortening of $(w_i, w_{i+1})$. Moreover, since $(w_i, w_{i + 1}) \in L^2$, then so is $(s, w_{i+1})$. 

    In particular, $L^2$ contains the path of links $(s, w_{i+1}), \ldots, (w_{k-1}, v)$ whose total cost is at most the cost of the path from $u$ to $v$. And hence, $L^3$ contains the shortening $(s, v)$ of $(u, v)$. 
\end{proof}

\begin{figure}[h]
\begin{centering}
\includegraphics[scale=0.7]{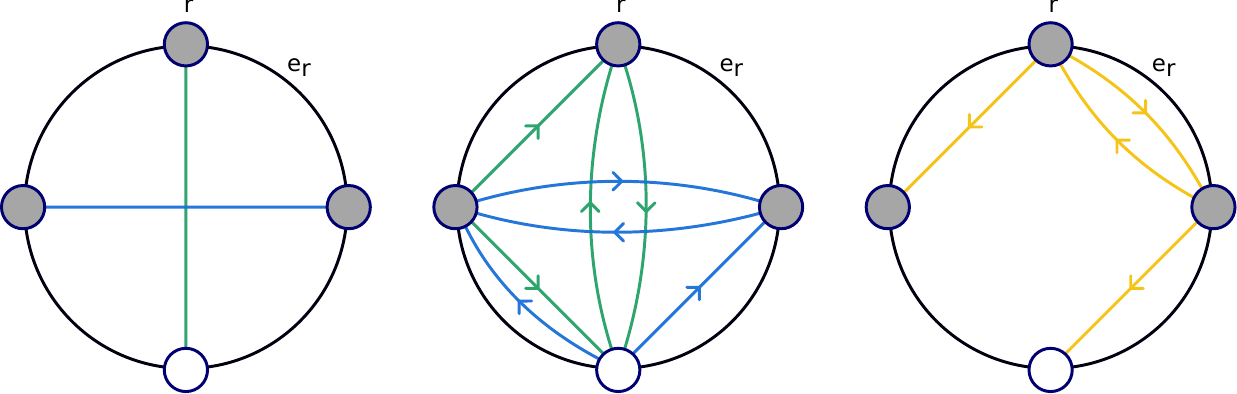} 
\caption{An example of a SRAP instance undergoing preprocessing steps to obtain a complete instance. The leftmost SRAP instance has two undirected links $\ell_1$ (green) and $\ell_2$ (blue) in $L$. There are no links added in $L^1$. The middle picture shows the directed links added in $L^2$, where the green arcs are shadows of $\ell_1$ and have cost $c(\ell_1)$, and the blue arcs are shadows of the $\ell_2$ with cost $c(\ell_2)$. Finally, the third picture shows the (undominated) directed links added in $L^3$ in yellow. Each of these arcs have cost $c(\ell_1) + c(\ell_2)$. The final completed SRAP instance contains all of these links.}\label{app:fig:complete-instance}
\end{centering}
\end{figure}

\section{A structured 2-approximate solution for SRAP}\label{app:sec:2approx}

Recall that for a rooted SRAP instance on ring $H$, the set of cuts to be covered is $$\mathcal{C} = \{C \subseteq V(H) \setminus r : |\delta_E(C)| = 2, C \cap R \neq \emptyset\}.$$ A set of directed links $\vecF$ is feasible if $\delta^-_{\vecF}(C) \geq 1$ for all $C \in \mathcal{C}$. 
In this section, we show that we can find a 2-approximate directed-link solution to SRAP which is only incident on terminals, and is highly structured. We call a directed solution which satisfies the conditions in Theorem~\ref{thm:2approx} an \textbf{$\boldsymbol{R}$-special} directed solution.

\twoApprox* 

\subsection{An $R$-special 2-approximate solution}
We now proceed with the 2-approximate $R$-special solution for SRAP. First, we make an existential claim about the existence of a 2-approximate directed solution which only touches terminals, then we show that we can find one efficiently.

In the proof of \Cref{app:lem:single-cycle}, we will first show that there is a 2-approximate directed solution which consists of a collection of directed cycles on nodes of the ring. We will then merge these cycles to obtain a directed solution which is a single directed cycle.

It is helpful to notice that an undirected cycle on nodes $A \subseteq V(H)$ covers the same set of dangerous ring-cuts that a directed cycle on $A$ covers, which is also the same set of dangerous ring-cuts that the hyper-link $A$ covers. Thus, the following lemma, proved in \Cref{app:sec:dropping_dir}, characterizes when a collection of such cycles is feasible. 

\begin{lemma}\label{lem:hyperlinksfeasible}
Suppose $(H, \mathcal{L}, R)$ is an instance of Hyper-SRAP with root $r$ and intersection graph $\Gamma$. Then $S \subseteq \mathcal{L}$ is feasible iff for each terminal $r' \in R$, there is a path between $r$ and $r'$ in the hyper-link intersection graph restricted to $S$.
\end{lemma}

We now proceed with the proof. For a node set $A \subseteq V(H)$, denote by $I_A \subseteq V(H)$ the nodes in the minimal interval containing $A$ which does not contain $e_r$. 

\begin{lemma}\label{app:lem:single-cycle}
    Given a complete instance of SRAP, let the optimal solution have cost $\OPT$. Then there exists a directed solution of cost at most $2\OPT$ whose links form a single directed cycle containing $R$.
\end{lemma}
\begin{proof}
    Let $(U^*, F^*)$ be a full component of the optimal SRAP solution. Then $(U^*, F^*)$ is a tree. Starting from an arbitrary vertex in $U^* \cap V(H)$, we take an Euler tour of this tree traversing each link exactly twice. This induces an ordering of the ring nodes which are visited during the tour $a_1, \ldots, a_k$, such that the undirected link set $F = \{(a_1,a_2), \ldots, (a_{k-1},a_k),(a_k,a_1)\}$ has cost at most $2c(F^*)$ by metric completion. 

    Now, consider the directed link set $\vecF = \{(a_1,a_2), \ldots, (a_{k-1},a_k),(a_k,a_1)\}$. The cost of $\vecF$ is at most $c(F)$ since it consists of shadows of links in $F$. Hence $c(\vecF) \leq 2c(F^*)$. Furthermore, since $\vecF$ forms a directed cycle on the nodes joined by $F^*$, any cut covered by $F^*$ is also covered by $\vecF$. 

    Repeating this for each full component in the optimal solution yields a directed solution with cost at most $2\OPT$. Furthermore, this directed solution consists of a collection of directed cycles on the nodes of the ring.

    We now exploit the following lemma to obtain a directed solution with cost at most $2\OPT$ consisting of a single directed (not necessarily simple) cycle.

\begin{lemma}[Cycle Merging Lemma]\label{app:lem:cycle-merging}
    Consider a SRAP instance on the ring $H = (V(H), E)$. Let $\vecF_S$ and $\vecF_A$ be two directed cycles of links on nodes $S, A \subseteq V(H)$, respectively, where $r \in S$. If $S$ and $A$ are intersecting as hyper-links, then there exists a directed link set $\vecF_{S \cup A}$ which forms a directed cycle on $S \cup A$ of cost at most $c(\vecF_S) + c(\vecF_A)$.
\end{lemma}
\begin{proof}
    First observe that if $A \cap S \neq \varnothing$, we simply define $\vecF_{S \cup A} = \vecF_S \cup \vecF_A$, and the statement of the lemma is satisfied.

    So, suppose that $A \cap S = \varnothing$. Denote the cycle links as $\vecF_S = \{(s_1, s_2), \ldots, (s_{k-1}, s_k), (s_k, s_1)\}$, with $s_1 = r$, and $\vecF_A = \{(a_1, a_2), \ldots, (a_{\ell-1}, a_\ell), (a_\ell, a_1)\}$. Consider the interval $I_A$ of $A$. Since $S$ and $A$ are intersecting (as hyper-links), and $S$ contains the root $r$ while $r \not \in I_A$, it must be the case that the cycle formed by $\vecF_S$ leaves $I_A$ at least once. Let $(s_i, s_{i+1}) \in \vecF_S$ be a directed link leaving $I_A$. 
    
    Now consider the interval $I^* = I_{\{s_i, s_{i+1}\}}$. We claim that $\vecF_A$ leaves $I^*$ at least once. Indeed, since $(s_i, s_{i+1})$ leaves $I_A$, some $a \in A$ is in $I^*$. But all of $A$ cannot be contained in $I^*$, since then $(s_i, s_{i+1})$ leaving $I_A$ would imply that $s_i \in A$, contradicting our assumption that $A \cap S = \varnothing$. Let $(a_j, a_{j+1}) \in \vecF_A$ be a directed link leaving $I^*$. 

    We now define $\vecF_{S \cup A}$ to be the directed links along the cycle
    $$
        s_1, \ldots, s_i, a_{j+1}, a_{j+2}, \ldots, a_{j}, s_{i+1}, \ldots, s_1.
    $$
    See \Cref{app:fig:cycle-merging} for an example. Observe that $(s_i, a_{j+1})$ is a shortening of $(a_j, a_{j+1})$ and $(a_j, s_{i+1})$ is a shortening of $(s_i, s_{i+1})$. Therefore $c(\vecF_{S \cup A}) \leq c(\vecF_S) + c(\vecF_A)$ as desired. Moreover, $\vecF_{S \cup A}$ forms a cycle on $S \cup A$.
    \begin{figure}[h]
    \begin{centering}
    \includegraphics[scale=0.7]{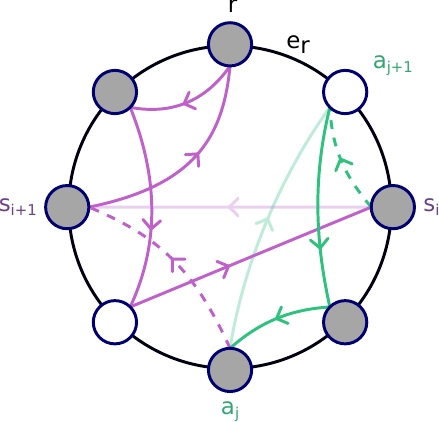} 
    \caption{An example of cycle merging, with $\vecF_S$ in purple and $\vecF_A$ in green. Note that $(s_i, s_{i+1})$ leaves the interval $I_A$, and $(a_j, a_{j+1})$ leaves the interval $I_{\{s_i, s_{i+1}\}}$. Dropping these two links, and adding the shortenings drawn as dotted arcs creates the single cycle $\vecF_{S \cup A}$.}\label{app:fig:cycle-merging}
    \end{centering}
    \end{figure}
\end{proof}

We can now use \Cref{app:lem:cycle-merging} to transform a directed solution that is a collection of directed cycles on ring nodes $A_1, \ldots, A_p$ into a single directed cycle on $A := \bigcup_{i=1}^p A_i$. Let $\vecJ$ denote the directed solution consisting of a collection of directed cycles on the nodes in the ring. Since $r$ is a terminal, it must be contained in some cycle $\vecF_{S_0}$ on $S_0 \subseteq V(H)$. We build a new directed solution starting with $\vecF_{S_0}$ by repeatedly apply \Cref{app:lem:cycle-merging}, cycle merging. In each step, choose a new cycle $\vecF_A$ from $\vecJ$ that has not yet been merged and such that $A$ is intersecting with $S_i$, and update $\vecF_{S_i}$ by merging: $\vecF_{S_{i+1}} \gets \vecF_{S_i \cup A}$. Once all cycles from $\vecJ$ are merged in this way, this yields a single directed cycle. It is feasible, since it is incident to all terminals, and its cost is at most that of $\vecJ$, which is at most $2\OPT$. 

It remains only to show that at every step, if there are un-merged cycles from $\vecJ$, then one of them is on a set of nodes which is intersecting with $S_i$. First, observe that $\vecF_{S_i}$ is a cycle on all of the nodes belonging to the cycles merged so-far, so it suffices to find some un-merged $\vecF_A$ in $\vecJ$ such that $A$ is intersecting with the nodes in one of the cycles merged so-far. 

Since $\vecJ$ is feasible, the corresponding collection of undirected cycles is also feasible. Hence, by \Cref{lem:hyperlinksfeasible}, the hyper-links on the vertices of the cycles in $\vecJ$ are connected in the hyper-link intersection graph. Now for any un-merged cycle $\vecF_{A'}$, look at the path in the hyper-link intersection graph from the hyper-link on $S_0$ to the hyper-link on $A'$. The first hyper-link along this path corresponding to an un-merged cycle $\vecF_A$ must be intersecting with a merged cycle, so is intersecting with $S_i$. 
\end{proof}

With the above lemma in hand, and due to our second metric completion step, we can now shortcut over the non-terminals in the cycle to obtain a 2-approximate directed solution which only touches terminals.

\begin{lemma}\label{app:lem:terminal-cycle}
    Given an instance of SRAP, if the optimal solution has cost $\OPT$, then there is a directed solution of cost at most $2\OPT$ whose links consist of a single directed cycle with node set $R$.
\end{lemma}
\begin{proof}
    This follows by taking the directed cycle solution $\vecF$ from \Cref{app:lem:single-cycle} and short-cutting through all non-terminal nodes. Since the third preprocessing step ensures that the directed links are metrically complete, replacing a path of directed links with a single link cannot increase the cost. Hence, the resulting solution has cost at most $c(\vecF) \leq 2\OPT$.
\end{proof}

Having shown that there exists a directed 2-approximate solution to any SRAP instance which is incident only on the terminals $R$, we now proceed to show how to compute one in polynomial time. In particular, we will prove that the optimal such solution can be found efficiently, and moreover, that the solution can be shortened to have additional structure.

Recall that the Weighted Ring Augmentation Problem (WRAP) is a special case of SRAP in which all nodes are terminals, and there are no nodes outside of the ring $H$. A directed solution $\vecF$ to WRAP is \textbf{non-shortenable} if it is feasible but deleting or strictly shortening any link $f \in \vecF$ results in an infeasible solution. Traub and Zenklusen proved that a non-shortenable directed solution to WRAP has a planar arborescence structure, and can be computed in polynomial time.

\begin{lemma}[Lemma 2.5 in~\cite{TZ2022new}]\label{app:lem:TZ2.5}
    A non-shortenable optimal directed solution to the WRAP problem can be found in polynomial time.
\end{lemma}

Reducing our problem to a WRAP instance and applying the above lemma yields the main theorem of this section:

\begin{proof}[Proof of \Cref{thm:2approx}]
    Consider a SRAP instance on the ring $H = (V(H), E)$ with terminals $R \subseteq V(H)$, root $r \in R$, and links $L$. We construct a new ring $H' = (R, E')$ on the terminals by iteratively replacing each non-terminal node in $V(H)$ with an edge between its two neighbors in $H$. We now create a WRAP instance on $H'$ by taking the subset $L'$ of links which are directed links incident only to $R$. 

    We apply \Cref{app:lem:TZ2.5} (Lemma 2.5 from~\cite{TZ2022new}) to get a non-shortenable optimal directed solution $\vecF$ to this WRAP instance. Since \Cref{app:lem:terminal-cycle} guarantees the existence of a directed 2-approximation which only touches terminals, this implies that $c(\vecF) \leq 2\OPT$. 

    Now consider $\vecF$ as a solution to the SRAP instance on $H$. Since $\vecF$ is non-shortenable when viewed as a solution to the WRAP instance on $H'$, Theorem 2.6 from~\cite{TZ2022new} shows that it satisfies the three conditions when viewed as a solution to $H$. It remains only to argue that $\vecF$ is actually feasible. This is immediate from the arborescence structure of $\vecF$: there is a path from $r$ to every terminal. 
\end{proof}

\section{A $(1+\ln{2} + \varepsilon)$-approximation for SRAP}\label{app:sec:relgreedyalgo}
In this section, we describe a polynomial time algorithm for SRAP which achieves an approximation ratio of $(1+\ln{2} + \varepsilon)$. 
Before writing the algorithm, we comment on a slight departure of our algorithm from the relative greedy algorithm for standard WRAP. 

Given any SRAP instance, we may convert it into an equivalent complete instance (see \Cref{app:sec:complete-instances}). We observe that for any directed link $f = (u, v)$ on the ring, there is a collection of undirected links whose coverage is at least that of $f$, and whose total cost is at most $c(f)$. We call this set $\kappa _f$. Indeed, if $f$ is a shadow of an undirected link $\ell$, then we may take $\kappa _f = \{\ell\}$.  Otherwise, $f$ arises from the metric completion of some directed path of links $(u = s_0,s_1), (s_1,s_2), \ldots, (s_{k-1},s_k = v)$. Each of these links in the directed path may themselves be shortenings of undirected links $\{\ell_1, \ldots, \ell_k\}$. In this case, $\kappa_f = \{\ell_1, \ldots, \ell_k\}$. Note in particular that $\kappa_f$ covers all of the dangerous ring-cuts that $f$ covers, so if $f$ is in some $R$-special solution $\vecF_0$, then $f \in \drop_{\vecF_0}(\kappa_f)$ (see \Cref{app:sec:dropping_dir} for a formal definition of $\drop$).

In order to show that our algorithm makes sufficient progress at each step, we must have that the cost of our mixed solution does not ever increase over the course of the algorithm. In the case of WRAP, this is immediate, since any directed link is a shadow of some single undirected link of the same cost, an $\alpha$-thin set. However, in our case, $\kappa_f$ may not be $\alpha$-thin, so we explicitly consider this as a separate case in each step of our algorithm. 
See \Cref{app:algo:localgreedy}.

\begin{algorithm}[t]
\caption{Relative greedy algorithm for SRAP}
\textbf{Input:} A complete instance of SRAP with graph $G = (V,E \dot \cup L)$, ring $H = (V(H),E)$, terminals $R \subseteq V(H)$ and $c: L \to \mathbb{R}$. Also an $\varepsilon > 0$.\\
\textbf{Output:} A solution $S \subseteq L$ with $c(S) \leq (1+\ln(2) + \varepsilon) \cdot c(\OPT)$.\\~

\begin{enumerate}
\item Compute a 2-approximate $R$-special directed solution $\vecF_0$ (\Cref{thm:2approx}).
\item Let $\varepsilon' := \frac{\varepsilon/2}{1+\ln 2 + \varepsilon/2}$ and $\gamma := 2^{\lceil{1/\varepsilon'}\rceil}$.
\item For each $A \subseteq V(H)$ where $|A| \leq \gamma$, compute the cheapest full component joining $A$ and denote the cost by $c_A$.
\item Create an instance of $\gamma$-restricted Hyper-SRAP on the ring $H = (V(H),E)$ with hyper-links $\mathcal{L} = \{\ell_A : A \subseteq V(H), |A| \leq \gamma\}$. Set the cost of hyper-link $\ell_A$ to be $c_A$.
\item Initialize $S_0 := \emptyset$
\item Let $\alpha:= 4 \lceil 4/\varepsilon \rceil$
\item While $\vecF_i \neq \emptyset$:
\begin{itemize}
    \item Increment $i$ by 1.
    \item Compute the $\alpha$-thin subset of hyper-links $Z_i \subseteq \mathcal{L}$ minimizing $\frac{c(Z_i)}{c(\drop_{\vecF_0}(Z_i) \cap \vecF_{i-1})}$. 
    \item If $\frac{c(Z_i)}{c(\drop_{\vecF_0}(Z_i) \cap \vecF_{i-1})} > 1$, then update $Z_i = \kappa_f$ for some $f \in \vecF_{i-1}$.
    \item Let $S_i := S_{i-1} \cup Z_i$ and let $\vecF_i := 
    \vecF_{i-1} \setminus \drop_{\vecF_0}(Z_i)$.
\end{itemize}
\item \textbf{Return} A SRAP solution with full components corresponding to the hyper-links in $S := S_i$.
\end{enumerate}
\label{app:algo:localgreedy}
\end{algorithm}

The proof of the following theorem follows from a standard analysis of the relative greedy algorithm as in~\cite{CN2013},~\cite{TZ2021}, and~\cite{TZ2022new}. We include it here for completeness. 

\begin{theorem}
    \Cref{app:algo:localgreedy} is a $(1+ \ln{2} + \varepsilon)$-approximation algorithm for SRAP. 
\end{theorem}
\begin{proof}
    First, we note that in each iteration of the algorithm, $|\vecF|$ will be reduced by at least 1. The initial value of $|\vecF|$ is at most $|R|-1$, so there are polynomially many iterations. By \Cref{app:thm:DPthm}, each iteration can be executed in polynomial time. Hence \Cref{app:algo:localgreedy} runs in polynomial time.

    The returned solution $S$ is feasible since the invariant that $S_i \cup \vecF_i$ is a feasible mixed solution is maintained throughout the algorithm.

    To complete the proof, we show that $S$ has cost at most $(1+\ln{2}+\varepsilon) \cdot c(\OPT)$. Denote by $\OPT_\gamma$ the optimal $\gamma$-restricted solution. Apply the decomposition theorem, \Cref{thm:decomposition}, to $\OPT_\gamma$ and the $R$-special solution $\vecF_0$. It gives a partition $\mathcal{Z}$ of $\OPT_\gamma$ into $\alpha$-thin parts, and some $Q \subseteq \vecF_0$ with $c(Q) \leq \sfrac{\varepsilon}{4} \cdot c(\vecF_0) \leq \sfrac{\varepsilon}{2} \cdot c(\OPT)$. Then observe that
    $$
        \frac{c(Z_i)}{c(\drop_{\vecF_0}(Z_i) \cap \vecF_{i-1})} \leq \min_{Z \in \mathcal{Z}} \frac{c(Z)}{c(\drop_{\vecF_0}(Z) \cap \vecF_{i-1})} \leq \frac{\sum_{Z \in \mc{Z}} c(Z)}{\sum_{Z \in \mc{Z}} c(\drop_{\vecF_0}(Z) \cap \vecF_{i-1})} \leq \frac{c(\OPT_\gamma)}{c(\vecF_{i-1}) - c(Q)},
    $$
    where the first inequality is by the choice of $Z_i$ in the algorithm, and the second and third follow from the statement of \Cref{thm:decomposition}. Since $f \in \drop_{\vecF_0}(\kappa_f)$, our choice of $Z_i$ implies that, $\sfrac{c(Z_i)}{c(\drop_{\vecF_0}(Z_i) \cap \vecF_{i-1})} \leq 1$. Therefore, since $\drop_{\vecF_0}(Z_i) \cap \vecF_{i-1} = \vecF_{i-1} \setminus \vecF_i$, we have
    $$
        c(Z_i) \leq \min\left\{1,  \frac{c(\OPT_\gamma)}{c(\vecF_{i-1}) - c(Q)}\right\} \cdot c(\vecF_{i-1} \setminus \vecF_i)  \leq \int_{c(\vecF_{i})}^{c(\vecF_{i-1})} \min\left\{1,  \frac{c(\OPT_\gamma)}{x - c(Q)}\right\}\,dx.
    $$
    Finally, sum over all iterations of the algorithm to get the cost of the output $S$:
    \begin{align*}
        c(S) &= \sum_i c(Z_i) \\
        &\leq \int_0^{c(\vecF_0)}\min\left\{1,  \frac{c(\OPT_\gamma)}{x - c(Q)}\right\}\,dx\\
        &= \int_0^{c(\OPT_\gamma) + c(Q)}1\,dx + \int_{c(\OPT_\gamma) + c(Q)}^{c(\vecF_0)}\frac{c(\OPT_\gamma)}{x - c(Q)}\,dx\\
        &\leq \left(1 + \frac{\varepsilon}{2}\right)\cdot c(\OPT_\gamma) + \ln\left( \frac{c(\vecF_0) - c(Q)}{c(\OPT_\gamma)} \right) \cdot c(\OPT_\gamma)\\
        &\leq \left(1 + \ln(2) + \frac{\varepsilon}{2}\right) \cdot c(\OPT_\gamma)
    \end{align*}
    And applying \Cref{lem:gamma_restriction} for our choice of $\varepsilon'$ and $\gamma$ from the algorithm, we have $c(\OPT_\gamma) \leq (1 + \varepsilon') \cdot c(\OPT)$. Hence, we get the desired bound
    $$
    c(S) \leq \left(1 + \ln(2) + \frac{\varepsilon}{2}\right) \cdot (1 + \varepsilon') \cdot c(\OPT) = \left(1 + \ln(2) + \varepsilon\right)\cdot c(\OPT). \qedhere
    $$
    
\end{proof}

\section{Dropping Directed Links}\label{app:sec:dropping_dir}
The main reason that we work with structured $R$-special directed solutions to SRAP is that it allows us to cleanly characterize when a directed link can be dropped after a collection of hyper-links are added to the solution. In this section, we recall the properties of an $R$-special directed solution and use them give two such characterizations.

\subsection{Dropping directed links via the hyper-link intersection graph}

Given an instance of Hyper-SRAP we can use the root $r$ and root-edge $e_r$ to define a notion of right and left along the ring. In particular, we imagine deleting the edge $e_r$ from the ring and consider the root to be the left-most node on the remaining path. The other node incident to $e_r$ is the right-most node in the ring. 

Consider a $R$-special directed solution $\vecF$ for the SRAP problem. Recall that $\vecF$ has the following properties:
\begin{enumerate}
    \item $\vecF$ is only incident on terminals $R$.
    \item $(R,\vecF)$ is an $r$-out arborescence.
    \item $(R, \vecF)$ is planar when $V(H)$ is embedded as a circle in the plane.
    \item For any $v \in R$, no two directed links in $\delta^+_{\vecF}(v)$ go in the same direction along the ring.
\end{enumerate}

Given an $R$-special directed SRAP solution $\vecF$, we associate to each cut $C \in \mathcal{C}$ a single link which is responsible for covering it, as in~\cite{TZ2022new}. In particular, an arc $\ell = (u,v) \in \vecF$ is \textbf{responsible} for covering a cut $C \in \mathcal{C}$ if $\ell$ enters $C$ and there is no other arc on the unique $r$--$u$ path in $(R,\vecF)$ which enters $C$. We denote the set of cuts for which a link $\ell \in \vecF$ is responsible by $\mathcal{R}_{\vecF}(\ell)$. Notice that $\mathcal{R}_{\vecF}(\ell) \neq \emptyset$ for all $\ell \in \vecF$, since if some directed link were not responsible for any cuts, then it could be deleted without affecting the feasibility of $\vecF$, but this is not true of any arc in an $R$-special directed solution.

We will show that every ring-dangerous cut $C \in \mathcal{C}$ has exactly one directed link responsible for it in an $R$-special solution $\vecF$. First we will need the following lemma which follows from the properties of $\vecF$.

\begin{lemma}\label{app:lem:descendants-interval}
    Let $\vecF$ be an $R$-special directed solution for SRAP. 
    \begin{enumerate}[(i)]
        \item For any $v \in R$, the set of descendants of $v$ in $(R, \vecF)$ is of the form $R \cap I$ for some ring interval $I \subseteq V(H)$. \label{app:lem:descendants-intervalp1}

        \item For $v_1, v_2 \in R$, the least common ancestor of $v_1$ and $v_2$ in $(R, \vecF)$ lies between $v_1$ and $v_2$. \label{app:lem:descendants-intervalp2}
    \end{enumerate}
\end{lemma}
This is an analogue to Lemma 4.8 from~\cite{TZ2022new}. In fact, this lemma follows immediately from Lemma 4.8 from~\cite{TZ2022new} by simply removing the Steiner nodes in the ring and viewing $\vecF$ as an $R$-special solution on a ring with only the terminals $R$ (c.f. the proof of \Cref{thm:2approx}).

\begin{lemma}\label{app:lem:responsible}
    For each cut $C \in \mathcal{C}$, there is exactly one directed link $\ell \in \vecF$ responsible for it.
\end{lemma}
\begin{proof}
    A cut $C \in \mathcal{C}$ clearly has at least one link responsible for it, since $\vecF$ enters all such cuts. On the other hand, no two links can both be responsible for $C$. This is because $C$ is an interval, so if $(u_1, v_1) \neq (u_2, v_2)$ are both responsible for $C$, then the least common ancestor $v$ of $v_1$ and $v_2$ is in $C$, by \Cref{app:lem:descendants-interval}\ref{app:lem:descendants-intervalp2}. Since $r \not \in C$, there must be some link in the $r$--$v$ path in $(R, \vecF)$ (and hence also on the $r$--$v_1$ path) entering $C$, contradicting that $(u_1,v_1)$ is responsible for $C$. 
\end{proof}

We can now define when a directed link will be dropped from an $R$-special directed solution $\vecF$. In particular, if a collection of hyper-links $K$ is added to the solution, then we will drop a directed link if and only if all cuts it is responsible for are covered by the hyper-links in $K$. Let $\delta_K(C)$ be the set of hyper-links in $K$ which cover $C \in \mathcal{C}$.

Formally, denote $$\drop_{\vecF}(K) = \{f \in \vecF : |\delta_{K}(C)| \geq 1 \text{ for all } C \in \mathcal{R}_{\vecF}(f)\}.$$

With this definition, if a collection of hyper-links $K$ is added to an $R$-special directed solution $\vecF$ to SRAP, and $\drop_{\vecF}(K)$ is removed, then the solution remains feasible as a mixed SRAP solution.

Our ultimate goal of this section is an alternate characterization of the set $\drop_{\vecF}(K)$. For this, we will use the notion of the hyper-link intersection graph. Recall that a pair of hyper-links $\ell_1$ and $\ell_2$ are \textbf{intersecting} if they share a vertex or there is a vertex in $\ell$ between two vertices in $\ell'$ and a vertex of $\ell'$ lies between two vertices of $\ell$.

Given an instance of Hyper-SRAP with ring $H = (V(H),E)$ and hyper-links $\mathcal{L}$, we define the hyper-link intersection graph $\Gamma$ as follows. For each hyper-link $\ell \in \mathcal{L}$ there is a node $v_\ell$. Two nodes $v_{\ell_1}$ and $v_{\ell_2}$ are adjacent in the hyper-link intersection graph if and only if $\ell_1$ and $\ell_2$ are intersecting hyper-links. For $K \subseteq \mathcal{L}$, $\Gamma(K)$ is the hyper-link intersection graph restricted to the hyper-links in $K$.

The following lemma and its proof are similar to the analogous statements for standard links (Lemma 5.4 in~\cite{TZ2022new}).
\begin{lemma}\label{app:lem:link-intersection-covering}
    Let $K \subseteq \mathcal{L}$ be a collection of hyper-links, and $C \in \mathcal{C}$ a dangerous ring-cut. Then $\delta_K(C) \neq \varnothing$ if and only if there is a path in $\Gamma(K)$ from a hyper-link containing some $v \in C$ to a hyper-link containing some $w \not \in C$. 
\end{lemma}
\begin{proof}
    If $\delta_K(C) \neq \varnothing$, then clearly such a path in $\Gamma(K)$ exists. In particular, any single hyper-link in $\delta_K(C)$ forms a path. 
    Conversely, suppose such a path exists in $\Gamma(K)$, but that $\delta_K(C) = \varnothing$ for contradiction. That is, any hyper-link in $K$ is either fully contained $C$ or fully contained in $V(H) \setminus C$. Since the path starts with a hyper-link containing $v \in C$, and ends with one containing $w \not \in C$, the path must contain some pair of intersecting hyper-links $\ell_1$ and $\ell_2$ with $\ell_1$ contained in $C$, while $\ell_2$ is contained in $V(H) \setminus C$. Hence,  clearly $\ell_1$ and $\ell_2$ do not share a vertex. Furthermore, since $C$ is an interval, any vertex lying between two vertices in $\ell_1$ must also lie in $C$, and hence not in $\ell_2$. This contradicts that $\ell_1$ and $\ell_2$ are intersecting. 
\end{proof}


\Cref{lem:hyperlinksfeasible} is an immediate corollary of \Cref{app:lem:link-intersection-covering}.

The notions of $v$-bad and $v$-good were used in~\cite{TZ2022new}, but the definitions need to be modified for our setting.

\begin{definition}
Let $v \in R$ and consider the maximal interval $I_v \subseteq V(H)$ containing $v$ such that $I_v$ does not contain a terminal which is a non-descendant of $v$ in $(R,\vecF)$. We say that the nodes in $I_v$ are \textbf{$\boldsymbol{v}$-bad}, and all nodes in $V(H) \setminus I_v$ are \textbf{$\boldsymbol{v}$-good}. 
\end{definition}

Observe that by \Cref{app:lem:descendants-interval}\ref{app:lem:descendants-intervalp1}, the interval of $v$-bad nodes actually contains all descendants of $v$ (but may also contain some Steiner nodes in the ring). In fact, we show that the family of $v$-bad intervals for $v \in R$ is a laminar family.

The following lemmas are analogous to Lemmas 5.6 and 2.10, respectively, in~\cite{TZ2022new}. The proofs are similar, but we include them here as our definitions of $v$-good and $v$-bad have changed. 

\begin{lemma}\label{app:lem:vbad-cut}
    A directed link $(u, v) \in \vecF$ is responsible for a cut $C \in \mathcal{C}$ if and only if $v \in C$ and all of the nodes in $C$ are $v$-bad.
\end{lemma}
\begin{proof}
    If $v \in C$ and all nodes of $C$ are $v$-bad, then $C \cap R$ contains only descendants of $v$. In particular, $u \not \in C$, so $(u, v)$ enters $C$, and no other link on the $r$--$v$ path in $\vecF$ can enter $C$. 

    Conversely, if $(u, v)$ is responsible for cut $C$, then it enters $C$, so $v \in C$. Moreover, for every $t \in C \cap R$, the $r$--$t$ path in $\vecF$ must have a link entering $C$. By \Cref{app:lem:responsible} we know that $(u, v)$ is the only link responsible for $C$, so every such $t$ must be a descendant of $v$. That is, $C \cap R$ are all descendants of $v$, and in particular, $C$ contains no terminals which are non-descendants of $v$. Hence, since $C$ is an interval, it is contained in the maximal interval containing no terminal non-descendants of $v$, which is precisely the set of $v$-bad nodes. 
\end{proof}

This yields the main result of this subsection, which provides a characterization of when a directed link $(u,v)$ can be dropped, namely when $v$ is connected to a $v$-good vertex through the hyper-link intersection graph. This is the criterion we will work with in \Cref{app:sec:decompthm}.

\begin{lemma}\label{app:lem: whencanIdrop}
For a collection of hyper-links $K$, a directed link $(u,v)$ is in $\drop_{\vecF}(K)$ if and only if $\Gamma(K)$ contains a path from a hyper-link containing $v$ to a hyper-link containing a $v$-good vertex $w$. 
\end{lemma}
\begin{proof}
    First suppose that there is such a path in $\Gamma(K)$. Then any cut $C \in \mathcal{R}_{\vecF}((u,v))$ for which $(u, v)$ is responsible cannot contain $w$, since by \Cref{app:lem:vbad-cut} it only contains $v$-bad nodes. Hence, by \Cref{app:lem:link-intersection-covering}, we have that $\delta_K(C) \neq \varnothing$. 

    Conversely, if $v$ is not connected to a $v$-good vertex by $K$ in the hyper-link intersection graph, then consider the set $K_v \subseteq V(H)$ of nodes reachable from $v$ via the hyper-link intersection graph of $K$. Let $I_{K_v}$ be the interval of these nodes. Then $I_{K_v} \in \mathcal{C}$, since $v \in I_{K_v}$, $r \not \in I_{K_v}$ (since $r$ is $v$-good) and it is a 2-cut. Moreover, since $K_v$ are all $v$-bad, then $I_{K_v}$ are as well. By \Cref{app:lem:vbad-cut}, this implies that $(u, v)$ is responsible for the cut $I_{K_v}$. We finish the proof by arguing that $K$ does not cover $I_{K_v}$, implying that $(u, v) \not \in \drop_{\vecF}(K)$. This is clear from the definition of the hyper-link intersection graph: if some hyperlink in $K$ covers $I_{K_v}$, then it must be intersecting with a hyperlink reachable from $v$ in the hyper-link intersection graph, thus contradicting that $K_v$ are all of the nodes reachable from $v$ in the hyper-link intersection graph of $K$. 
\end{proof}

\subsection{Dropping directed links via the artifical link extension}

We now describe how we can add artificial directed links to the $R$-special solution $\vecF_0$ to obtain a $V(H)$-special solution $\vecF'$. This brings our situation closer to the standard WRAP problem, and in particular allows us to give an alternative characterization of which directed links can be dropped which will be useful for the dynamic programming algorithm in \Cref{app:sec:dp}. 

We first show that the $v$-bad intervals form a laminar family.

\begin{lemma}
    Suppose $\vecF_0$ is an $R$-special solution, and let $I_v \subseteq V(H)$ denote the $v$-bad interval. Then the family of all $v$-bad intervals $\mathcal{F} = \{I_v : v \in R\}$ is laminar.
\end{lemma} 
\begin{proof}

    Consider two terminals $u$ and $v$ with $u$-bad interval $I_{u}$ and $v$-bad interval $I_v$. By \Cref{app:lem:descendants-interval}\ref{app:lem:descendants-intervalp2}, the least common ancestor (in $\vecF_0$) $w$ of $u$ and $v$, lies between them. If $w$ is not $u$ or $v$, then it is a non-descendant of both of them, so neither $I_u$ nor $I_v$ can contain $w$. Since $u \in I_u$ and $v \in I_v$, these intervals do not intersect. 
    
    Otherwise, $w$ is equal to either $u$ or $v$; suppose $w = u$ without loss of generality. But then all non-descendants of $u$ are non-descendants of $v$, so $I_v \subseteq I_u$. Thus, the family is laminar.
\end{proof}

We now describe how we add artificial directed links to the $R$-special solution $\vecF_0$ to obtain a $V(H)$-special solution $\vecF'$. We will then use $\vecF'$ to define a notion of least common ancestor in $\vecF_0$ that is defined for any subset of ring nodes.

For each terminal $v \in R$, define the set of nodes $S_v \subseteq V(H)$ to consist of all nodes $u$ such that $I_v$ is the minimal interval of $\mathcal{F}$ which contains $u$. Notice that $S_v$ is an interval containing exactly one terminal, namely $v$, and potentially Steiner nodes to the left and right of $v$. Suppose $S_v = S^L_v \cup S^R_v \cup \{v\}$ where $S^L_v$ and $S^R_v$ are the (possibly empty) sets of Steiner nodes in $S_v$ to the right and left of $v$ respectively. Suppose $S^L_v = \{s^L_0, \ldots, s^L_k\}$ ordered right to left along the ring, and $S^R_v = \{s^R_0, \ldots, s^R_m\}$ ordered left to right along the ring. We now add the set of artificial directed links $\{(v,s^L_0), (s^L_0,s^L_1), \ldots, (s^L_{k-1}, s^L_k)\} \cup \{(v,s^R_0), (s^R_0,s^R_1), \ldots, (s^R_{m-1}, s^R_m)\}$. See \Cref{app:fig:artificial-links}.

\begin{figure}[h]
\begin{centering}

\includegraphics[scale=0.7]{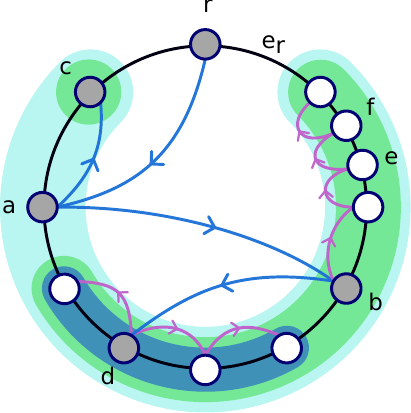} 
\caption{An example of an $R$-special solution with $R = \{r,a,b,c,d\}$ and its extension to an artificial $V(H)$-special solution. The artificial links are purple. The $r$-bad interval is always $V(H)$. The $a$-bad interval is shown in cyan. The $b$-bad and $c$-bad intervals are green, and the $d$-bad interval is dark blue. Note that $\overline \lca(\{e,f\}) = e$ and $\overline \lca(\{e,f,d\}) = b$.}\label{app:fig:artificial-links}

\end{centering}
\end{figure}

When these artificial links are added to $\vecF_0$, we obtain an associated $V(H)$-special solution $\vecF'$. Observe that for a terminal $v$, the set of descendants of $v$ in $\vecF'$ (including $v$ itself) is exactly its $v$-bad interval $I_v$. We use this to extend the least common ancestor function with respect to $\vecF_0$ to all subsets of nodes (even though the arborescence $\vecF_0$ only touches terminals). 

\begin{definition}[Extended Least Common Ancestor Function]
    Suppose $\vecF_0$ is an $R$-special solution and $A \subseteq V(H)$. Then the least common ancestor of $A$ with respect to $\vecF_0$ is $$\overline \lca(A) := \lca_{\vecF'}(A),$$ where $\lca_{\vecF'}(A)$ is the least common ancestor of $A$ in the artificial solution $\vecF'$.
\end{definition}

Notice that the above definition coincides with the standard least common ancestor function with respect to $\vecF_0$ on sets $A$ which consist only of terminals.

Finally, we state the key \Cref{app:lem:dropLCA}. For a set of hyper-links $S$, let $V(S) := \bigcup_{\ell \in S} \ell$.

\begin{lemma}\label{app:lem:dropLCA}
    If $\vecF_0$ is an $R$-special directed solution to SRAP and $S \subseteq \mathcal{L}$ is a collection of hyper-links which form a connected component in the hyper-link intersection graph, then $$\drop_{\vecF_0}(S) = \Big(\bigcup_{v \in V(S)} \delta_{\vecF_0}^-(v)\Big) \setminus \delta_{\vecF_0}^-( \overline{\lca}(V(S))).$$
\end{lemma}

\begin{proof}
First, observe that if $V(S) \cap R = \varnothing$, then both sets in the lemma statement are empty, and the lemma is true. So assume that $V(S)$ contains some terminal. 

Since $\delta^-_{\vecF_0}(v) = \emptyset$ whenever $v \not \in R$, we consider some vertex $v \in V(S) \cap R$. Suppose that $v \neq \overline{\lca}(V(S))$. Since $v \in V(S)$, there must be some vertex $u$ in $V(S)$ which is not a descendant of $v$ with respect to $\vecF'$. Since the set of descendants of $v$ in $\vecF'$ is exactly the $v$-bad interval $I_v$, we have that $u$ is a $v$-good vertex. Thus, by \Cref{app:lem: whencanIdrop}, we have that $\delta_{\vecF_0}^-(v) \in \drop_{\vecF_0}(S)$, since $S$ connects $v$ to a $v$-good vertex via the hyper-link intersection graph.

On the other hand, if $v = \overline{\lca}(V(S))$, then all other vertices in $V(S)$ are descendants of $v$ in $\vecF'$. Again, the set of descendants of $v$ in $\vecF'$ is its $v$-bad interval $I_v$, so all vertices in $V(S)$ are $v$-bad, and so by \Cref{app:lem: whencanIdrop}, $\delta_{\vecF_0}^-(v)$ is not contained in $\drop_{\vecF_0}(S)$.
\end{proof}

\section{The Decomposition Theorem for Hyper-SRAP}\label{app:sec:decompthm}
In this section, we prove the following theorem.

\decompositionTheorem*

We follow the approach in~\cite{TZ2022new} which proves the result when all hyper-links have size 2 and $R = V(H)$. We will first partition the hyper-links in $S$ into a collections of festoons whose spans form a laminar family. We will then construct a dependency graph whose nodes correspond to festoons, which will allow us to partition the links of $S$ into the desired $\alpha$-thin pieces. In the general hyper-link setting, the festoons are composed of hyper-links rather than links of size 2.

We begin by defining festoons in the context of hyper-links. The interval of a hyper-link $\ell$ is denoted $I_\ell$ and is the set of vertices in the interval between the leftmost vertex and the rightmost vertex of $\ell$. We say that hyper-links $\ell$ and $\ell'$ are \textbf{crossing} if their intervals intersect and neither is a subset of the other. Recall that two hyper-links are \textbf{intersecting} if they share a vertex, or there is a vertex in $\ell$ between two vertices in $\ell'$ and a vertex of $\ell'$ lies between two vertices of $\ell$ (when the vertices are viewed in the left to right order along the ring). Notice that if two hyper-links are crossing, then they also are intersecting.

\begin{definition}
    A \textbf{festoon} is a set of hyper-links $X \subseteq \mathcal{L}$ which can be ordered $\ell_1, \ldots, \ell_p$ such that $\ell_i$ and $\ell_{i+1}$ are crossing for $i \in \{1, \ldots, p-1\}$ and $I_{\ell_i}$ and $I_{\ell_j}$ are disjoint unless $|i-j| = 1$. 
\end{definition}

Notice that whether a set of hyper-links is a festoon only depends on their hyper-link intervals. Let $\mathcal{C}' = \{C \subseteq V(H) : |\delta_E(C)| = 2, r \not \in C\}$ be the family of minimum cuts of the ring $H$. One of the key properties of festoons is that no minimum cut is covered by more than 4 hyper-links in a festoon. 
\begin{lemma}
    Let $X$ be a festoon and $C \in \mathcal{C}'$. Then $|\{\ell \in X : \ell \text{ covers } C\}| \leq 4$.
\end{lemma}
\begin{proof}
Suppose $X = {\ell_1, \ldots, \ell_p}$ is a festoon of hyper-links and $C \in \mathcal{C}'$. Then $C$ is an interval between vertices say $u$ and $v$ in $V(H)$, where $u$ is the left endpoint of $C$ and $v$ is the right endpoint.

We will show that there are at most two hyper-links in $\delta_X(C)$ which contain vertices to the left of $u$. A symmetric argument shows that there can be at most two hyper-links in $\delta_X(C)$ containing a vertex to the right of $v$. Together, these imply that there are at most 4 hyper-links in $\delta_X(C)$.

Suppose there are 3 hyper-links covering $C$ which contain vertices to to the left of $u$. Since they cover $C$, each of these hyper-links must also contain some vertex in $C$. Thus, their hyper-link intervals all contain the vertex $u$. But this is impossible, as the definition of festoons requires that $I_{\ell_i}$ and $I_{\ell_j}$ are disjoint unless $|i-j| = 1$.  
\end{proof}

\begin{definition}[$\alpha$-thin set of hyper-links]
    A set of hyper-links $K$ is $\alpha$-thin if there exists a maximal laminar subfamily $\mathcal{D}$ of $\mathcal{C}'$ such that for every cut $C \in \mathcal{D}$, we have $|\{\ell \in K : \ell \text{ covers } C\}| \leq \alpha$.
\end{definition}

Given a festoon $X$, let $I_X$ denote the festoon interval which is the interval from the left-most vertex in the festoon to the right-most.
We can partition $S$ into a collection of festoons $\mathcal{X}$, so that the festoon intervals form a laminar family. To do this, we iteratively construct the partition by choosing a festoon with the largest interval in each iteration and adding it to the partition.
We now prove that if the set $S$ is partitioned in this way, it yields a partition $\mathcal{X}$ such that the set of festoon intervals form a laminar family.

\begin{lemma}
    The festoon intervals $\{I_X : X \in \mathcal{X}\}$ form a laminar family.
\end{lemma}
\begin{proof}
Suppose that $X, Y \in \mathcal{X}$ are such that $I_X$ and $I_Y$ cross. We also assume without loss of generality that $X$ was added the festoon family before $Y$ was. 

Let $\ell_1, \ldots, \ell_p$ be the hyper-links in the festoon $X$, numbered according to the festoon order. We will assume $X$ is to the left of $Y$. Since they cross, there is some hyper-link in $Y$ whose left endpoint is in $I_X$. But then this hyper-link would have been included in $X$ to form a longer interval. 
\end{proof}

Given the laminar structure of the festoon intervals, we can define a partial order on the festoons in $\mathcal{X}$. In particular, we will say that $X \prec Y$ if $I_X \subseteq I_Y$.

\begin{definition}
    We say that two festoons $X$ and $Y$ are \textbf{tangled} if some hyper-link in $X$ is intersecting with some hyper-link of $Y$.
\end{definition}

Suppose that $X_1, \ldots, X_p$ is a sequence of festoons such that $X_1$ contains a vertex $v \in R \setminus r$, $X_i$ and $X_j$ are tangled iff $|i-j| = 1$, and $X_p$ contains a $v$-good vertex, but no festoons $X_i$ for $i < p$ contain a $v$-good vertex. Then by \Cref{app:lem: whencanIdrop}, $\{X_1,\ldots,X_p\}$ is a minimal collection of festoons such that the directed link $(u,v) \in \vecF$ can be dropped. We will choose for each $v \in R$, a minimal collection of festoons $\mathcal{X}_v = \{X_1, \ldots, X_p\}$ as above so that $|I_{X_p}|$ is minimized. This will allow us to construct the dependency graph with the desired properties.

We now turn to defining the dependency graph $(\mathcal{X},A)$.
It will have a vertex for each festoon in $\mathcal{X}$, and its arcs are obtained by inserting for each directed link $(u,v) \in \vecF$ a directed path corresponding to a minimal set of festoons as above. Formally, for each $v \in R \setminus r$, there will be a directed path $P_v$ consisting of the festoons in $\mathcal{X}_v$ where $(X_i,X_{i-1}) \in A$ for $i \in \{2, \ldots p\}$. 

It will be helpful to consider a subgraph of the dependency graph which is constructed only from paths $P_v$ where $v \in U$ for some set of terminals $U \subseteq R \setminus r$. This is called the \textbf{$U$-dependency graph}. 

We will use that the dependency graph is a branching, which follows from the minimality in its construction and the partial order defined on festoons. Clearly, this implies that the $U$-dependency graph is also a branching for any $U \subseteq R \setminus r$.

\begin{lemma}\label{app:lem:branching}
    The dependency graph $(\mathcal{X},A)$ is a branching. That is, the in-degree of every node $X \in \mathcal{X}$ is at most 1 and it contains no cycles.
\end{lemma}
\begin{proof}
    Since $(X,Y) \in A$ implies $Y \prec X$, it is clear that $(\mathcal{X},A)$ contains no directed cycles. We now show that each festoon $X$ has in-degree at most 1. Suppose for the sake of contradiction that $(Y,X) \in A$ and $(Z,X) \in A$ where $Y \neq Z$. Since the arcs in $A$ arise from a collection of directed paths, it must be the case that $(Y,X) \in P_v$ and $(Z,X) \in P_w$ for distinct vertices $v,w \in R \setminus r$.
    
    We will use the following key property, used in~\cite{TZ2022new}, which continues to hold in our setting: for any pair of terminals $v,w \in R \setminus r$, either $w$ is $v$-good or $v$ is $w$-good. Hence, we assume without loss of generality that $w$ is $v$-good. 
    
    Since $(Z,X) \in P_w$, we have that $w$ is contained in the festoon interval $I_X$ of $X$. Since the $v$-good vertices form an interval, $X$ must contain a hyper-link which contains a $v$-good vertex. But by \Cref{app:lem: whencanIdrop}, this means that $Y$ is not necessary for a directed link entering $v$ to be dropped, contradicting $(Y,X) \in P_v$.
\end{proof}

We can bound the thinness of the components of the dependency graph with the following lemmas, which are direct extensions of analogous lemmas in~\cite{TZ2022new}.  Let $\alpha$ be some positive integer, and $U \subseteq R \setminus r$.

\begin{lemma}\label{app:lem:tangled_and_thin}
    If $\mathcal{X}'$ is a connected component of the $U$-dependency graph and the set $$\{Z \in \mathcal{X}' : X \preceq Z, \text{ } X \text{ and } Z \text{ are tangled}\}$$ has cardinality at most $\alpha$ for every festoon $X \in \mathcal{X}'$, then $\mathcal{X}'$ is $4(\alpha + 1)$-thin.
\end{lemma}
\begin{proof}
    The proof is the same as the proof of Lemma 7.16 in~\cite{TZ2022new}. Note that Lemma 7.8 in~\cite{TZ2022new} holds for festoons of hyper-links as well as standard festoons.
\end{proof}

\begin{lemma}\label{app:lem:tangled_are_ancestors}
    Suppose $X$ and $Y$ are festoons in the same connected component of the $U$-dependency graph. If $X$ and $Y$ are tangled, then they have an ancestor-descendant relationship in the $U$-dependency graph. 
\end{lemma}
\begin{proof}
    Again, the proof is a direct extension of Lemma 7.17 in~\cite{TZ2022new}.
\end{proof}


\begin{lemma}\label{app:lem:key_lemma_dependencygraph}
    Let $(\mathcal{X}',A')$ be a connected component of the $U$-dependency graph. If $$|\{u \in V : P \cap P_u \neq \emptyset\}| \leq \alpha$$ for every directed path $P \subseteq A'$ then the collection of links which are contained in festoons of $\mathcal{X}'$ is $4(\alpha+1)$-thin.
\end{lemma}
\begin{proof}
First, notice that for every $X \in \mathcal{X}$, there are at most $\alpha$ festoons in $\mathcal{X}'$ which are at least $X$ in the partial ordering. This is because $X$ can be tangled with at most one festoon from each $P_u$, which follows from minimality as well as \Cref{app:lem:tangled_are_ancestors}. Hence, by \Cref{app:lem:tangled_and_thin}, $\mathcal{X}'$ is $4(\alpha+1)$-thin as desired.
\end{proof}

Having shown that the dependency graph is a branching, the property given in \Cref{app:lem:key_lemma_dependencygraph} is enough to replicate the same argument which was introduced by Traub and Zenklusen in~\cite{TZ2021} for the Weighted Tree Augmentation Problem to prove the decomposition theorem. We break the dependency graph up into pieces, such that each corresponds to a collection of hyper-links which is $\alpha$-thin, while only destroying a small fraction of the sets $P_u$. We reproduce the proof below.

\begin{proof}[Proof of \Cref{thm:decomposition}]
Let $q := \lceil \frac{1}{\varepsilon} \rceil$. We will construct an arc labeling for each connected component $(\mathcal{X}',A')$ of the
$(R \setminus r)$-dependency graph $(\mathcal{X},A)$, which is a branching by \Cref{app:lem:branching}. The arcs in the same set $P_u$ will receive the same label. 

We define the labeling inductively as follows. For each directed path $P_u$ which begins at the root of the arborescence $(\mathcal{X}',A')$, we set the labels of the arcs in this path to be 0. For a directed path $P_u$ which begins at a node $X$ which is not the root, we set the label of the arcs in $P_u$ to be $j+1$, where $j$ is the label of the unique arc entering $X$. We perform the labeling in this fashion for each connected component of the dependency graph to obtain a labeling of all arcs in $A$

Now, let $Q_i \subseteq \vecF$ be the set of directed links $(u,v)$ such that $P_v$ received label $i$, where $i \in \{0, \ldots, q-1\}$. This is a partition of $\vecF$ into $q$ parts. Hence, the average cost of the sets $Q_i$ is $c(\vecF) / q$, implying that the cheapest set among $Q_0, \ldots, Q_{q-1}$ has cost at most $c(\vecF) / q$. Let this cheapest part be denoted by $Q \subseteq \vecF$. 

We define $U \subseteq R$ to be those terminals which are not entered by a directed link in $Q$, and consider the $U$-dependency graph. The $U$-dependency graph is obtained by deleting from $(\mathcal{X},A)$ all arcs with label $i$ for some $i \in \{0, \ldots, q-1\}$. Hence, each of its connected components satisfies the hypothesis of \Cref{app:lem:key_lemma_dependencygraph} where $\alpha = q-1$, implying that the set of hyper-links in the festoons of each component is $4q$-thin. This yields our partition of the hyper-links of $S$ into parts where each part is $4\lceil \frac{1}{\varepsilon} \rceil$-thin.  

Finally, for each directed link $(u,v) \in \vecF \setminus Q$ there is some connected component of the $U$-dependency graph which contains all the arcs in $P_v$. Hence, $(u,v)$ is droppable by adding all the hyper-links in the festoons of this component. Since, $c(Q) \leq \varepsilon \cdot c(\vecF)$, we have the desired property. 
\end{proof}

\section{Dynamic Programming to find the best $\alpha$-thin component}\label{app:sec:dp}
In this section, we prove that, given an $R$-special directed solution $\vecF_0$, and a subset $\vecF \subseteq \vecF_0$, the $\alpha$-thin collection of $\gamma$-restricted hyper-links $K$ which minimizes the ratio $$\frac{c(K)}{c(\drop_{\vecF_0}(K) \cap \vecF)}$$ can be found in polynomial time.

As is now standard following the results of~\cite{TZ2021}\cite{TZ2022new}\cite{RZZ2022}, we will focus on maximizing $\slack_\rho(K)$ defined as $$\slack_\rho(K) := \rho \cdot c(\drop_{\vecF_0}(K) \cap \vecF) - c(K).$$

Notice that deciding whether the ratio $\frac{c(K)}{c(\drop_{\vecF_0}(K) \cap \vecF)}$ is smaller than a fixed $\rho^*$ is equivalent to deciding whether $\slack_{\rho^*}(K)$ is greater than 0. Thus, if we can efficiently find a maximizer of the slack function for any given $\rho$, then we can use binary search to obtain our desired result. 

It will be convenient for the results in \Cref{app:sec:1.5approx} to prove a more general result allowing different cost functions on the two terms of the slack function. 

\begin{theorem}\label{app:thm:DPthm}
    Given an instance of SRAP, let $\vecF_0$ be an $R$-special directed solution. Let $\tilde{c}: \vecF_0 \to \mathbb{R}_{\geq 0}$ be a cost function on $\vecF_0$. Then a maximizer of $\Tilde{c}(\drop_{\vecF_0}(K)) - c(K)$ over all $\alpha$-thin subsets of $\gamma$-restricted hyper-links can be computed in polynomial time. 
\end{theorem}

Given \Cref{app:thm:DPthm}, we can maximize $\slack_\rho$ by setting $\tilde c(\ell) = \rho \cdot c(\ell)$ for $\ell \in \vecF \cap \vecF_0$ and 0 otherwise. We will prove the above theorem by dynamic programming. 

Traub and Zenklusen prove this optimization theorem for the WRAP problem in~\cite{TZ2022new}. The proof in our setting is an extension of their methods. The key differences in our context are twofold. First, we are working with an $R$-special directed solution $\vecF_0$ which is incident only on the terminals $R$, whereas in WRAP, the arborescence contains all the nodes of the ring. Secondly, we must extend their techniques to hyper-links, rather than standard links containing pairs of vertices of the ring. 

To handle this, we will add artificial links to the $R$-special solution $\vecF_0$ to obtain a $V(H)$-special solution $\vecF'$, which allows us to extend the least common ancestor function to all subsets of ring nodes, as described in \Cref{app:sec:dropping_dir}. This allows us to leverage \Cref{app:lem:dropLCA}, which is the analogue of Lemma 2.11 in~\cite{TZ2022new}, and is crucial in computing new table entries from already computed ones. Finally, we exploit the fact that we are only working with hyper-links of size at most $\gamma$, implying that there are polynomially many hyper-links available.

For a ring $H = (V(H),E)$, let the set of 2-cuts not containing the root be denoted by $\mathcal{C}' := \{C \subseteq V(H) : |\delta_E(C)| = 2, r \not \in C\}$. Recall that a collection of hyper-links $K$ is $\alpha$-thin if there exists a maximal laminar subfamily $\mathcal{D}$ of $\mathcal{C'}$ such that there are at most $\alpha$ hyper-links in $K$ which cover $C$ for each $C \in \mathcal{D}$. In the following, we will build up a solution by considering subproblems defined on intervals of the ring. As such we need a definition of a hyper-link set which is an $\alpha$-thin with respect to an interval $C \subseteq V(H)$.

\begin{definition}
    A collection of hyper-links $K$ is $(\alpha,C)$-thin if there exists a maximal laminar subfamily $\mathcal{D}$ of $\mathcal{C}'$ on ground set $C$, such that for each $\bar C \in \mathcal{D}$, there are at most $\alpha$ hyper-links from $K$ which cover $\bar C$.
\end{definition}

We will use the notation $\delta(C)$ to denote the set of hyper-links which cover the cut $C \in \mathcal{C}'$. We maintain a table $T$ of polynomial size with a table entry $T[C,B, \mathcal{T}, \phi,\psi]$ where:

\begin{itemize}
    \item $C \in \mathcal{C}'$,
    \item $B \subseteq \delta(C)$ is of size at most $\alpha$,
    \item $\mathcal{T}$ is a partition of $B$,
    \item $\phi: \mathcal{T} \to V$,
    \item $\psi: \mathcal{T} \to \{0,1\}$.
\end{itemize}

Note that there are $|V(H)|^2$ choices for $C$. Since there are at most $|V(H)|^\gamma$ possible $\gamma$-restricted hyper-links, there are at most $|V(H)|^{\gamma \alpha}$ choices for $B$. Finally, if $\alpha$ is a constant then there are constantly many choices for $\mathcal{T}$ and $\psi$, and polynomially many choices for $\phi$. Thus, the overall dimensions of the table are polynomial.

We now describe how to interpret a subproblem corresponding to a table entry $T[C,B,\mathcal{T},\phi,\psi]$. A subset $S \subseteq \mathcal{L}$ of hyper-links realizes this table entry if:

\begin{itemize}
    \item The hyper-links in $S$ contain some vertex of $C$,
    \item $\delta_S(C) = B$,
    \item $\mathcal{T}$ consists of the non-empty sets $S_i \cap B$, where $S_1, \ldots, S_q$ are the connected components of $S$ in the hyper-link intersection graph,
    \item For each set $S_i \cap B \in \mathcal{T}$, we have $\phi(S_i \cap B) = \overline\lca(V(S_i))$,
    \item For each set $S_i \cap B \in \mathcal{T}$, we have $\psi(S_i \cap B) = 1$ if and only if $\phi(S_i \cap B) \in V(S_i)$.
\end{itemize}

Thus, the table entry $T[C,B,\mathcal{T},\phi,\psi]$ contains the maximizer and maximum value of $$\tilde c\left(\drop_{\vecF_0}(S) \cap \bigcup_{v \in C} \delta^{-}_{\vecF_0} (v)\right) - c(S)$$ over all $(\alpha,C)$-thin collections of hyper-links $S$ which realize this table entry.

Traub and Zenklusen show how to compute the table entry $T[C,B,\mathcal{T},\phi,\psi]$ from previously computed ones in polynomial time. At a high level, we will enumerate over all possible table entries whose solutions can be combined to yield a solution to $T[C,B,\mathcal{T},\phi,\psi]$. This is done by guessing a partition of $C$ into two neighboring cuts $C_1$ and $C_2$ from $\mathcal{C}'$ (notice that there are at most $|V(H)|$ choices for this partition), and choices of $B_1$ and $B_2$ which are compatible with each other and also respect the $(\alpha,C)$-thinness. In particular, we must have $\delta_{B_2}(C_1) = \delta_{B_1}(C_2)$ and $|\delta_{B_1 \cup B_2}(C)| \leq \alpha$. Finally, we must have $\mathcal{T}_i$, $\phi_i$ and $\psi_i$ interact in such a way as to yield a solution to $T[C,B,\mathcal{T},\phi,\psi]$ when their solutions are combined. 

Suppose $(C_1,B_1,\mathcal{T}_1,\phi_1,\psi_1)$ and $(C_2,B_2,\mathcal{T}_2,\phi_2,\psi_2)$ are table entries such that $C_1$ and $C_2$ are adjacent intervals, $\delta_{B_2}(C_1) = \delta_{B_1}(C_2)$ and $|\delta_{B_1 \cup B_2}(C_1 \cup C_2)| \leq \alpha$. Then the merger of these table entries is defined as $(C,B,\mathcal{T},\phi,\psi)$, where $C = (C_1 \cup C_2)$, $B = \delta_{B_1 \cup B_2}(C)$, and $\mathcal{T}$, $\phi$, and $\psi$ are defined as follows:

Consider the graph with vertex set $B_1 \cup B_2$ where $x$ and $y$ are adjacent if either: $x$ and $y$ are intersecting hyper-links, $x$ and $y$ are in $B_1$ and in the same set in the partition $\mathcal{T}_1$, or $x$ and $y$ are in $B_2$ and in the same set of the partition $\mathcal{T}_2$. Then $\mathcal{T}$ is the partition of $B_1 \cup B_2$ corresponding to the connected components of this graph. For any $T \in \mathcal{T}$ which by definition is equal to the union of some parts $(T^1_1, T^1_2, \ldots, T^1_{q_1})$ from $\mathcal{T}_1$ and some parts $(T^2_1, T^2_2, \ldots, T^2_{q_2})$ from $\mathcal{T}_2$, define $\phi(T) = \overline\lca(\phi_1(T^1_1), \ldots, \phi_1(T^1_{q_1}), \phi_2(T^2_1), \ldots, \phi_2(T^2_{q_2}))$. Finally, let $\psi(T) = 1$ if there exists some $T^i_j$ with $\psi_i(T_j) = 1$ and $\phi_i(T^i_j) = \phi(T)$. With this definition, if $S_1$ is a set of hyper-links which realizes $(C_1,B_1,\mathcal{T}_1,\phi_1,\psi_1)$ and $S_2$ realizes $(C_2,B_2,\mathcal{T}_2,\phi_2,\psi_2)$, then $S_1 \cup S_2$ realizes their merger.

Thus, to compute the optimal value for table entry $T[C,B,\mathcal{T},\phi,\psi]$, we can enumerate over pairs $(C_1,B_1,\mathcal{T}_1,\phi_1,\psi_1)$ and $(C_2,B_2,\mathcal{T}_2,\phi_2,\psi_2)$ whose merger is $(C,B,\mathcal{T},\phi,\psi)$. Let $\pi(C,B,\mathcal{T},\phi,\psi)$ denote the optimal value for this table entry. Then 
\begin{align*}
\pi(C,B,\mathcal{T},\phi,\psi) = \max \Bigg\{&\pi(Q_1) + \pi(Q_2) + c(B_1 \cap B_2) + \sum_{u \in U_{Q_1,Q_2}} \tilde c(\delta^-_{\vecF_0}(u)) :\\
&Q_1 = (C_1,B_1,\mathcal{T}_1,\phi_1,\psi_1), \text{ and } \\
&Q_2 = (C_2,B_2,\mathcal{T}_2,\phi_2,\psi_2) \text{ have merger } (C,B,\mathcal{T},\phi,\psi)\Bigg\},
\end{align*}
where $U_{Q_1,Q_2} = \{\phi_i(T) : T \in \mathcal{T}_i \text{ with } \psi_i(T) = 1 \text{ and } \phi_i(T) \in C_i, \text{ for } i = 1,2\} \setminus \{\bigcup_{T \in \mathcal{T}} \phi(T)\} $. Furthermore, the optimal solution $S$ to $(C,B,\mathcal{T},\phi,\psi)$ will be $S_1 \cup S_2$ where $S_1$ is the optimizer of $Q_1^*$ and $S_2$ is the optimizer of $Q_2^*$ for the best choice of $Q_1^*$ and $Q_2^*$ in the above maximization.

The overall solution to the problem will be found in some table entry with $C = V(H) \setminus r$. Since the table has polynomially many entries, and each can be filled in polynomial time, this proves \Cref{app:thm:DPthm}.

\section{A $(1.5+\varepsilon)$-approximation for $k$-SAG}\label{app:sec:1.5approx}
\subsection{The local search algorithm}
In this section, we show how to use the local search framework introduced in~\cite{TZ2022} to give a $(1.5 + \varepsilon)$-approximation algorithm for SRAP, in the case that $R = V(H)$. By the arguments in \Cref{app:sec:srap-reductions}, this yields a $(1.5+\varepsilon)$ approximation algorithm for $k$-SAG. See \Cref{app:algo:localsearch}.

\begin{algorithm}[h]
\caption{Local search algorithm for SRAP}
\textbf{Input:} A complete instance of SRAP with graph $G = (V,E \dot \cup L)$, ring $H = (R,E)$, and $c: L \to \mathbb{R}$. Also a constant $1 \geq \varepsilon > 0$.\\
\textbf{Output:} A solution $S \subseteq L$ with $c(S) \leq (1.5 + \varepsilon)\OPT$.\\~

\begin{enumerate}
\item Compute an arbitrary SRAP solution $S \subseteq L$. Construct witness sets $W_f$ for each $f \in S$ so that $\vecF := \bigcup_{f \in S} W_f$ is an $R$-special directed solution.
\item Let $\varepsilon' := \frac{\varepsilon/2}{1.5 + \varepsilon/2}$ and $\gamma := 2^{\lceil{1/\varepsilon'}\rceil}$.
\item For each $A \subseteq R$ where $|A| \leq \gamma$, compute the cheapest full component joining $A$ and denote the cost by $c_A$.
\item Create an instance of $\gamma$-restricted Hyper-SRAP on ring $H = (R,E)$ with hyper-links $\mathcal{L} = \{\ell_A : A \subseteq R, |A| \leq \gamma\}$. Set the cost of hyper-link $\ell_A$ to be $c_A$.
\item Let $\alpha := \lceil 8/\varepsilon \rceil$
\item\label{step:loc-search} Iterate the following as long as $\Phi(S)$ decreases in each iteration by at least a factor $(1 - \frac{\varepsilon}{12n})$:
\begin{enumerate}[(i)]
    \item Compute the $\alpha$-thin subset of hyper-links $Z \subseteq \mathcal{L}$ maximizing $\bar c(\drop_{\vecF}(Z)) - 1.5 \cdot c(Z)$, where $\vecF := \bigcup_{f \in S} W_f$. 
    \item Update the witness sets by replacing each $W_f$ with $W_f \setminus \drop_{\vecF}(Z)$.
    \item Update $S$ by adding all links of full components corresponding to the hyper-links in $Z$.
    \item\label{step:shorten-to-special} Shorten directed links in $\vecF$ to obtain an $R$-special solution. If $W_f = \emptyset$ for some $f \in S$, then remove $f$ from $S$.
\end{enumerate}
\item \textbf{Return} $S$.
\end{enumerate}
\label{app:algo:localsearch}
\end{algorithm}

The algorithm begins by computing an arbitrary SRAP solution $S$. For each link $f \in S$, we will construct a witness set $W_f$ which initially consists of two directed links. In each iteration of the algorithm, we add a collection of links to the solution along with their associated witness sets, and drop directed links from other witness sets. If the witness set of a link $f$ becomes empty then $f$ is removed from $S$. Throughout the algorithm, we maintain that $\vecF := \bigcup_{f \in S} W_f$ is a feasible directed solution. The algorithm terminates when there is no local move which substantially improves the potential function.

We now define how the initial witness sets are constructed from our arbitrary starting solution $S$. 
Suppose a link $f$ is in a full component $(U',S')$, where $U \subseteq V$ and $S' \subseteq S$. Consider the eulerian tour traversing each link in $S'$ exactly twice. This tour induces an ordering on the ring nodes in $U' \cap V(H)$, say $\{a_1, \ldots, a_k\}$. Suppose that $f \in S'$ is traversed on the Euler subpath from $a_{i}$ to $a_{i+1}$ and on the subpath from $a_j$ to $a_{j+1}$ (where we take $a_{k+1} := a_1$). Then the witness set $W_f$ will consist of the two directed links $(a_i, a_{i+1})$ and $(a_j,a_{j+1})$. 

Note that $\vecF := \bigcup_{f \in S} W_f$ is a feasible directed solution. We now update the directed links in the witness sets by iteratively shortening links as long as $\vecF$ remains feasible. By Theorem 2.6 in~\cite{TZ2022new}, at the end of this shortening process, $\vecF$ is an $R$-special directed solution. 

Now, we define a potential function $\Phi$ defined on a solution $S$ along with its witness sets.

$$\Phi(S) := \sum_{f: |W_f| = 1} c(f) + \frac{3}{2}\sum_{f: |W_f| = 2} c(f).$$

We also define a weight function for the directed links in witness sets. For a directed link $u$ in some $W_f$, we define $$\bar c(u) := \sum_{f: u \in W_f} \frac{c(f)}{|W_f|}.$$

Our algorithm then proceeds in step 6 by choosing an $\alpha$-thin collection of hyper-links maximizing $\bar c(\drop_{\vecF}(Z)) - 1.5 \cdot c(Z)$. The links in the full components corresponding to these hyper-links are added to our solution $S$, while the directed links in $\drop_{\vecF}(Z)$ are removed from witness sets. Finally, witness sets are constructed for the new undirected links which were added to the solution, and all directed links in witness sets are shortened so that $\vecF$ becomes an $R$-special solution once more. 

Notice that \Cref{app:thm:DPthm} implies that the maximization step in step 6 of \Cref{app:algo:localsearch} can be performed in polynomial time. The following lemma, whose proof is identical to Lemma 9.3 in~\cite{TZ2022new}, shows that the algorithm will terminate after polynomially many iterations. Let $\OPT$ denote the cost of the optimal augmentation. 

\begin{lemma}
\Cref{app:algo:localsearch} runs for at most $$\ln{\Big(\frac{1.5\cdot c(S)}{\OPT}\Big)} \cdot \frac{6|R|}{\varepsilon}$$ iterations where $S$ is the initial SRAP solution.
\end{lemma}

Finally, the solution returned has cost at most $(1.5+\varepsilon)\OPT$.

\begin{lemma}
    At the end of \Cref{app:algo:localsearch}, we have $c(S) \leq (1.5 + \varepsilon)\OPT$.
\end{lemma}

The proof of the above lemma uses the same potential function analysis which is has been used in~\cite{RZZ2022}\cite{TZ2022new}\cite{TZ2022}, so we do not reproduce it again here. These yield the main result of this section:

\begin{theorem}
\Cref{app:algo:localsearch} is a $(1.5+\varepsilon)$-approximation algorithm for SRAP when $R = V(H)$.
\end{theorem}

\begin{corollary}
    There is a polynomial time $(1.5+\varepsilon)$-approximation algorithm for $k$-SAG.
\end{corollary}

\subsection{Necessity of $R = V(H)$}

We remark that in part \ref{step:shorten-to-special} of step \ref{step:loc-search} of \Cref{app:algo:localsearch}, it is crucial that $V(H) = R$. In the case that $R = V(H)$, any feasible directed SRAP solution can be iteratively shortened to obtain an $R$-special solution of at most the cost. This follows from Theorem 2.6 in~\cite{TZ2022new}.

However, this does not hold if $R \neq V(H)$ (see \Cref{app:fig:cant-make-special}). We are only able to obtain an initial 2-approximate $R$-special solution in~\Cref{thm:2approx} because we are initially given a feasible directed solution with a good deal of structure: namely one which consists of a collection of directed cycles on the nodes of the ring. In order to improve the approximation for general SRAP from $1+\ln(2) + \varepsilon$ to $1.5+\varepsilon$, one possibility is to argue that any directed solution which arises from the collection of witness sets during the course of the algorithm can be made $R$-special. We leave this for future work. 

\begin{figure}[h!]
    \centering

    \includegraphics[scale = 0.7]{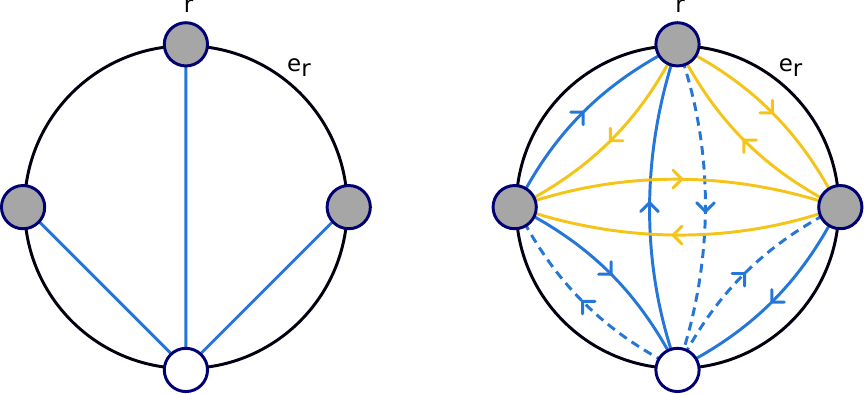}
    \caption{An example of a SRAP instance on the left, where all three undirected links have cost 1. The grey nodes are terminals. The completed instance is shown on the right, where the blue links have cost 1 and the orange links have cost 2 (only the directed links are shown). The dashed blue arcs form a feasible directed solution of cost 3, but there is no $R$-special solution of cost at most 3.}
    \label{app:fig:cant-make-special}
\end{figure}

\bibliography{bib2doifinal}

\end{document}